\documentclass[12pt]{article}
\usepackage{amssymb,amsmath,amsthm,mathptmx,makeidx,graphicx,url}
\usepackage{comment}
\DeclareMathOperator{\KS}{\textit{C}\hspace*{0.5pt}}
\let\leq=\leqslant
\let\le=\leqslant
\let\geq=\geqslant
\let\ge=\geqslant
\DeclareMathOperator{\bin}{\mathrm{bin}}
\DeclareMathOperator{\KP}{\mathit{K}}

\includecomment{problems}
\newtheorem{theorem}{Theorem}

\makeatletter
\renewcommand*\l@section[2]{%
  \ifnum \c@tocdepth >\z@
    \addpenalty\@secpenalty
    \setlength\@tempdima{1.5em}%
    \begingroup
      \parindent \z@ \rightskip \@pnumwidth
      \parfillskip -\@pnumwidth
      \leavevmode %\bfseries
      \advance\leftskip\@tempdima
      \hskip -\leftskip
      #1\nobreak\leaders\hbox to 0.5em{\hss.\hss}\hfil \nobreak\hb@xt@\@pnumwidth{\hss #2}\par
    \endgroup
  \fi}
\makeatother

\begin{document}
\title{Around Kolmogorov complexity:\\ basic notions and results}
\author{Alexander Shen\thanks{LIRMM (Montpellier), on leave from IITP RAS (Moscow), %\hbox to 3.5cm{\hss}
\texttt{alexander.shen@lirmm.fr}}}
\date{}

\maketitle

\begin{abstract}
Algorithmic information theory studies description complexity and randomness and is now a well known field of theoretical computer science and mathematical logic. There are several textbooks and monographs devoted to this theory~\cite{li-vitanyi,calude,nies,downey-hirschfeldt,vus} where one can find the detailed exposition of many difficult results as well as historical references. However, it seems that a short survey of its basic notions and main results relating these notions to each other, is missing. 
This report attempts to fill this gap and covers the basic notions of algorithmic information theory: Kolmogorov complexity (plain, conditional, prefix), Solomonoff universal a priori probability, notions of randomness (Martin-L\"of randomness, Mises--Church randomness), effective Hausdorff dimension. We prove their basic properties (symmetry of information, connection between a priori probability and prefix complexity, criterion of randomness in terms of complexity, complexity characterization for effective dimension) and show some applications (incompressibility method in computational complexity theory, incompleteness theorems). It is based on the lecture notes of a course at Uppsala University given by the author~\cite{uppsala-notes}.
\end{abstract}

\section{Compressing information}

Everybody is familiar with compressing/decompressing
programs such as \verb|zip|, \verb|gzip|, \verb|compress|,
\verb|arj|, etc. A compressing program can be applied to an arbitrary
file and produces a ``compressed version'' of that file. If we
are lucky, the compressed version is much shorter than the
original one. However, no information is lost: the decompression
program can be applied to the compressed version to get the
original file.\footnote{Imagine that a software company advertises a compressing program and claims that this program can compress \emph{every} sufficiently long file to at most $90\%$ of its original
size. Why wouldn't you buy this program?}

How is it possible? A compression program tries to find
some regularities in a file which allow it to give a
description of the file than is shorter than the file itself;
the decompression program  then reconstructs the file using this
description.

\section{Kolmogorov complexity}

The Kolmogorov complexity
may be roughly described as ``the
compressed size''. However, there are some differences. Instead of files (byte sequences) we consider bit strings (sequences of zeros and 
ones). The principal difference is that in the framework of
Kolmogorov complexity we have no \emph{compression} algorithm
and deal only with \emph{decompression} algorithms.

Here is the definition. Let $U$ be an algorithm whose inputs
and outputs are binary strings. Using $U$ as a decompression
algorithm, we define the complexity $\KS_U(x)$ of a binary string
$x$ with respect to $U$ as follows:
 $$
    \KS_U(x) = \min\{ |y| \colon U(y)=x\}
 $$
(here $|y|$ denotes the length of a binary string $y$). In other
words, the complexity of $x$ is defined as the length of the
shortest description of $x$ if each binary string $y$ is
considered as a description of $U(y)$

Let us stress that $U(y)$ may be defined not for all $y$, and
there are no restrictions on the time necessary to compute
$U(y)$. Let us mention also that for some $U$ and $x$ the set of descriptions in the definition of $\KS_U$ may be empty; we assume that
$\min(\emptyset)= +\infty$ in this case.

\section{Optimal decompression algorithm}

The definition of $\KS_U$ depends on $U$. For the trivial
decompression algorithm $U(y)=y$ we have $\KS_U(x)=|x|$.
One can try to find better decompression algorithms, where
``better'' means ``giving smaller complexities''. However,
the number of short descriptions is limited: There is less than
$2^n$ strings of length less than $n$. Therefore, for every fixed
decompression algorithm the number of strings whose complexity is
less than $n$ does not exceed $2^n-1$. One may conclude that
there is no ``optimal'' decompression algorithm because we can
assign short descriptions to some string only taking them away
from other strings. However, Kolmogorov made a simple but crucial
observation: there is \emph{asymptotically optimal} decompression
algorithm.

\textbf{Definition}
An algorithm $U$ is asymptotically not worse than an algorithm $V$
if $\KS_U(x) \le \KS_V(x) + C$ for come constant $C$ and for all $x$.

\begin{theorem}
There exists an decompression
algorithm $U$ which is asymptotically not worse
than any other algorithm $V$.
\end{theorem}

Such an algorithm is called \emph{asymptotically optimal}.
The complexity $\KS_U$ with respect to an asymptotically optimal
$U$ is called \emph{Kolmogorov complexity}. The Kolmogorov
complexity of a string $x$ is denoted by $\KS(x)$. (We assume that
some asymptotically optimal decompression algorithm is fixed.)
Of course, Kolmogorov complexity is defined only up to $O(1)$
additive term.

The complexity $\KS(x)$ can be interpreted as the amount of
information in $x$ or the ``compressed size'' of $x$.

\section[The construction of the optimal decompression algorithm]{The construction of optimal\\ decompression algorithm}

The idea of the construction is used in the so-called
``self-extracting archives''. Assume that we want to send a
compressed version of some file to our friend, but we are not
sure he has the decompression program. What to do? Of course, we
can send the program together with the compressed file. Or we
can append the compressed file to the end of the program
and get an executable file which will be applied to its own
contents during the execution (assuming that the operating system allows to append arbitrary data to the end of an executable file).

The same simple trick is used to construct an universal
decompression algorithm~$U$. Having an input string $x$, the
algorithm $U$ starts scanning $x$ from left to right until it
founds some program $p$ written in a fixed programming language
(say, Pascal) where programs are self-delimiting, so the end of
the program can be determined uniquely. Then the rest of $x$ is
used as an input for $p$, and $U(x)$ is defined as the output of
$p$.

Why $U$ is (asymptotically) optimal? Consider  another
decompression algorithm $V$. Let $v$ be a (Pascal) program which
implements $V$. Then
 $$
        \KS_U (x) \le \KS_V(x) + |v|
 $$
for arbitrary string $x$. Indeed, if $y$ is a $V$-compressed version of
$x$ (i.e., $V(y)=x$), then $vy$ is $U$-compressed version of $x$
(i.e., $U(vy)=x$) and is only $|v|$ bits longer.

\section{Basic properties of Kolmogorov complexity}

\begin{theorem}
\hspace*{ 1em}\par

\textbf{\textup{(a)}} $\KS(x) \le |x| + O(1)$.

\textbf{\textup{(b)}} The number of $x$ such that $\KS(x)\le n$ is equal to $2^n$ up to a bounded factor separated from zero.

\textbf{\textup{(c)}} For every computable function $f$ there exists a
   constant $c$ such that $$\KS(f(x))\le \KS(x)+c$$ \textup(for every $x$ such that
   $f(x)$ is defined\textup).

\textbf{\textup{(d)}} Assume that for each natural $n$ a finite set $V_n$ containing no more than $2^n$ elements is given. Assume that the relation $x\in V_n$ is enumerable, i.e., there is an algorithm which produces the (possibly infinite) list of all pairs $\langle x,n\rangle$ such that $x\in V_n$. Then there is a constant $c$ such that all elements of $V_n$ have complexity at most $n+c$ (for every $n$).

\textbf{\textup{(e)}} The ``typical'' binary string of length $n$ has complexity close to $n$: there exists a constant $c$ such that for every $n$ more than $99\%$ of all strings of length $n$ have complexity in-between $n-c$ and $n+c$.

\end{theorem}

\begin{proof}
 (a) The asymptotically optimal decompression algorithm
$U$ is not worse that the trivial decompression algorithm $V(y)=y$.

(b) The number of such $x$ does not exceed the number of their
compressed versions, which is limited by the number of all binary
strings of length not exceeding $n$, which is bounded by $2^{n+1}$.
On the other hand, the number of $x$'s such that $K(x)\le n$ is
not less than $2^{n-c}$  (here $c$ is the constant from (a)),
because all strings of length $n-c$ have complexity not exceeding $n$.

(c) Let $U$ be the optimal decompression algorithm used in the
definition of $\KS$. Compare $U$ with decompression algorithm
$V\colon y\mapsto f(U(y))$:
 $$
        \KS_U(f(x)) \le \KS_V(f(x))+O(1) \le \KS_U(x) +O(1)
 $$
(each $U$-compressed version of $x$ is a $V$-compressed version
of $f(x)$).

(d) We allocate strings of length $n$ to be compressed
versions of strings in $V_n$ (when a new element of $V_n$
appears during the enumeration, the first unused string of
length $n$ is allocated). This procedure provides a decompression
algorithm $W$ such that $\KS_W(x)\le n$ for every $x\in V_n$.

(e) According to (a), all strings of length $n$
have complexity not exceeding $n+c$ for some $c$. It remains to
mention that the number of strings whose complexity is less than
$n-c$ does not exceed the number of all their descriptions,
i.e., strings of length less than $n-c$. Therefore, for $c=7$
the fraction of strings having complexity less than $n-c$ among
all the strings of length $n$ does not exceed $1\%$.

\end{proof}

\begin{problems}

\subsection*{Problems}

\leavevmode

1. A decompression algorithm $D$ is chosen in such a way that
   $\KS_D(x)$ is even for every string $x$. Could $D$ be optimal?

2. The same question if $\KS_D(x)$ is a power of $2$ for every $x$.

3. Let $D$ be the optimal decompression algorithm. Does it
   guarantee that $D(D(x))$ is also an optimal decompression
   algorithm?

4. Let $D_1,D_2,\dots$ be a computable sequence of decompression
   algorithms. Prove that
        $
   \KS(x) \leqslant \KS_{D_i}(x)+2\log i +O(1)
        $
   for all $i$ and $x$ (the constant in $O(1)$ does not depend
   on $x$ and $i$).

5.$^*$ Is it true that $\KS(xy)\leq \KS(x)+\KS(y)+O(1)$ for all $x$
   and $y$?
   
\end{problems}   

\section{Algorithmic properties of $\KS$}

\begin{theorem}
The complexity function $\KS$ is not computable; moreover,
every computable lower bound for $\KS$ is bounded from above.
\end{theorem}

\begin{proof}
Assume that some partial function $g$ is a computable lower bound for $\KS$, and $g$ is not bounded from above.
Then for every $m$ we can effectively find a string $x$ such that
$\KS(x)>m$ (indeed, we should compute in parallel $g(x)$ for all strings $x$ until we
find a string $x$ such that $g(x)>m$). Now consider the function
$$
   f(m) = \hbox{the first string $x$ such that $g(x)>m$}
$$
Here ``first'' means ``first discovered'' and $m$ is a natural
number written in binary notation; by our assumption, such $x$ always exists, so $f$ is a total computable function. By construction, $\KS(f(m))>m$; on the other hand, $\KS(f(m))\le \KS(m)+O(1)$. But $K(m)\le |m|+O(1)$, so we conclude that $m \le |m|+O(1)$ which is impossible (the left-hand side is a natural number, the right-hand side---the length of its binary representation).
\end{proof}

This proof is a formal version of the well-known Berry paradox about ``the smallest natural number which cannot be defined by twelve English words'' (the quoted sentence defines this number and contains exactly twelve words).

\medskip
The non-computability of $\KS$ implies that any optimal decompression algorithm $U$ is not everywhere defined (otherwise $\KS_U$ would be computable). It
sounds like a paradox: If $U(x)$ is undefined for some $x$ we
can extend $U$ on $x$ and let $U(x)=y$ for some $y$ of large complexity; after that $\KS_U(y)$ becomes smaller (and all other values of $\KS$ do not change). However, it can be done for one $x$ or for finite number of $x$'s but we cannot make $U$ defined everywhere and keep $U$ optimal at the same time.

\section{Complexity and incompleteness}

The argument used in the proof of the last theorem may be used
to obtain an interesting version of G\"odel first incompleteness
theorem. This application of complexity theory was invented and advertised by G.~Chaitin.

Consider a formal theory (like formal arithmetic or formal set
theory). It may be represented as a (non-terminating) algorithm
which generates statements of some fixed formal language;
generated statements are called {\it theorems}. Assume that the
language is rich enough to contain statements saying that ``complexity
of 010100010 is bigger than 765'' (for every bit string and every
natural number). The language of the formal arithmetic satisfies
this condition as well as the language of the formal set theory. Let
us assume also that all theorems of the considered theory are true.

\begin{theorem}
There exists a constant $c$ such that all the theorems
of type ``$\KS(x)>n$'' have $n<c$.
\end{theorem}

\begin{proof}
Indeed, assume that it is not true. Consider the following
algorithm $\alpha$: For a given integer $k$, generate all the
theorems and look for a theorem of type $\KS(x)>s$ for some $x$
and some $s$ greater than~$k$. When such a theorem is found, $x$
becomes the output $\alpha(s)$ of the algorithm. By our
assumption, $\alpha(s)$ is defined for all $s$.

All theorems are supposed to be true, therefore $\alpha(s)$ is a
bit string whose complexity is bigger than $s$. As we have seen,
this is impossible, since $K(\alpha(s))\le K(s)+O(1) \le
|s|+O(1)$ where $|s|$ is the length of the binary representation
of $s$. 
\end{proof}

(We may also use the statement of the preceding theorem instead
of repeating the proof.)

This result implies the classical G\"odel theorem (it says that there are true unprovable statements), since there exist strings of arbitrarily high complexity.

A constant $c$ (in the theorem) can be found explicitly if we fix a formal
theory and the optimal decompression algorithm and for most
natural choices does not exceed --- to give a rough estimate ---
$100,000$. It leads to a paradoxical situation: Toss a coin
$10^6$ times and write down the bit string of length
$1,000,000$. Then with overwhelming probability its complexity
will be bigger than $100,000$ but this claim will be unprovable
in formal arithmetic or set theory. 

\section{Algorithmic properties of $\KS$ (continued)}

\begin{theorem}
The function $\KS(x)$ is upper semicomputable, i.e., $\KS(x)$
can be represented as $\lim\limits_{n\to\infty}k(x,n)$ where
$k(x,n)$ is a total computable function with integer values
and
       $$
   k(x,0)\ge k(x,1) \ge k(x,2) \ge\ldots
       $$
\end{theorem}

Note that all values are integers, so for every $x$ there exists some $N$ such that $k(x,n)=\KS(x)$ for all $n>N$.

Sometimes upper semicomputable functions are called \emph{enumerable from above}.

\begin{proof}
Let $k(x,n)$ be the complexity of $x$ if we restrict by
$n$ the computation time used for decompression. In other words,
let $U$ be the optimal decompression algorithm used in the
definition of $\KS$. Then $k(x,n)$ is the minimal $|y|$ for all
$y$ such that $U(y)=x$ and the computation time for $U(y)$ does
not exceed $n$. 
\end{proof}

(Technical correction: it can happen (for small $n$) that our
definition gives $k(x,n)=\infty$. In this case we let
$k(x,n)=|x|+c$ where $c$ is chosen in such a way that $\KS(x)\le
|x|+c$ for all $x$.)

\section{An encodings-free definition of complexity}

The following theorem provides an ``encodings-free'' definition
of Kolmogorov complexity as a minimal function $K$ such that $K$
is upper semicomputable and $|\{x \mid K(x)< n\}|=O(2^n)$.

\begin{theorem}
Let $K(x)$ be an upper semicomputable function such that
$|\{x \mid K(x)< n\}|\le M\cdot2^n$ for some constant $M$ and for all
$n$. Then there exists a constant $c$ such that $\KS(x)\le
K(x)+c$ for all~$x$.
\end{theorem}

\begin{proof}
This theorem is a reformulation of one of the statements
above. Let $V_n$ be the set of all strings such that $K(x)<n$.
The binary relation $x\in V_n$ (between $x$ and~$n$) is
enumerable. Indeed, $K(x)=\lim k(x,m)$ where $k$ is a total
computable function that is decreasing as a function of~$m$.
Compute $k(x,m)$ for all $x$ and $m$ in parallel. If it happens
that $k(x,m)<n$ for some $x$ and $m$, add $x$ into the
enumeration of $V_n$. (The monotonicity of $k$ guarantees that
in this case $K(x)<n$.) Since $\lim k(x,m)=K(x)$, every
element of $V_n$ will ultimately appear.

By our assumption $|V_n|\le M\cdot 2^n$. Therefore we can allocate
strings of length $n+c$ (where $c=\lceil \log_2M\rceil$) as
descriptions of elements of $V_n$ and will not run out of
descriptions. In this way we get a decompression algorithm $D$ such that
$\KS_D(x)\le n+c$ for $x\in V_n$. Since $K(x)<n$ implies
$\KS_D(x)\le n+c$ for all $x$ and $n$, we have
$\KS_D(x)\le K(x)+1+c$ and $\KS(x)\le K(x)+c$ for some other $c$
and all $x$. 
\end{proof}

\section{Axioms of complexity}

It would be nice to have a list of ``axioms'' for Kolmogorov
complexity that determine it uniquely (up to a bounded additive term).
The following list shows one of the possibilities.

\begin{itemize}
\item{A1}  (Conservation of information)
  For every computable (partial) function $f$ there exists a
  constant $c$ such that $K(f(x))\le K(x)+c$
  for all $x$ such that $f(x)$ is defined.
\item{A2}  (Enumerability from above)
  Function $K$ is enumerable from above.
\item{A3}  (Calibration)
  There are constants $c$ and $C$ such that the cardinality of
  set $\{x\mid K(x) < n\}$ is between $c\cdot2^n$
  and $C\cdot2^n$.
\end{itemize}

\begin{theorem}
Every function $K$ that satisfies A1--A3 differs from $\KS$
only by $O(1)$ additive term.
\end{theorem}

\begin{proof}
Axioms A2 and A3 guarantee that $\KS(x)\le K(x)+O(1)$. We
need to prove that $K(x)\le \KS(x)+O(1)$.

First, we prove that $K(x)\le |x|+O(1)$.

Since $K$ is enumerable from above, we can generate
strings $x$ such that $\mathsf{K}(x)<n$. Axiom A3 guarantees
that we have at least $2^{n-d}$ strings with this property for
some $d$ (which we assume to be an integer). Let us stop generating
them when we have already $2^{n-d}$ strings $x$ such that
$K(x)<n$; let $S_n$ be the set of strings generated in
this way. The list of all elements in $S_n$ can be obtained by
an algorithm that has $n$ as input; $|S_n|=2^{n-d}$ and
$K(x)<n$ for each $x\in S_n$.

We may assume that $S_1\subset S_2\subset S_3\subset\ldots$ (if
not, replace some elements of $S_i$ by elements of $S_{i-1}$
etc.). Let $T_i$ be equal to $S_{i+1}\setminus S_i$. Then $T_i$
has $2^{n-d}$ elements and all $T_i$ are disjoint.

Now consider a computable function $f$ that maps elements of
$T_n$ onto strings of length $n-d$. Axiom A1 guarantees then
that $K(x)\le n+O(1)$ for every string of length $n-d$.
Therefore, $K(x)\le |x|+O(1)$ for all $x$.

Let $D$ be the optimal decompression algorithm from the definition of $\KS$. We apply A1 to the function $D$. If $p$ is a shortest description for $x$, then $D(x)=p$, therefore $K(x)=K(D(p))\le K(p)+O(1) \le |p|+O(1)=\KS(x)+O(1).$

\end{proof}

\begin{problems}
\subsection*{Problems}

\leavevmode

1. If $f\colon\mathbb{N}\to\mathbb{N}$ is a computable bijection,
   then $\KS(f(x))=\KS(x)+O(1)$. Is it true if $f$ is a (computable)
   injection (i.e., $f(x)\ne f(y)$ for $x\ne y$)? Is it true if
   $f$ is a surjection (for every $y$ there is some $x$ such that
   $f(x)=y$)?

2. Prove that $\KS(x)$ is ``continuous'' in the following sense:
   $\KS(x0)=\KS(x)+O(1)$ and $\KS(x1)=\KS(x)+O(1)$.

3. Is it true that $\KS(x)$ changes at most by a constant if we
   change the first bit in $x$? last bit in $x$? some bit in $x$?

4. Prove that $\KS(\overline{x}01\bin(\KS(x)))$ (a string $x$ with
   doubled bits is concatenated with 01 and the binary
   representation of its complexity $\KS(x)$) equals $\KS(x)+O(1)$.
\end{problems}

\section{Complexity of pairs}

Let
        $$
x,y\mapsto[x,y]
        $$
be a computable function that maps pairs of strings into
strings and is an injection (i.e., $[x,y]\ne[ x',y']$ if $x\ne
x'$ or $y\ne y'$). We define complexity $\KS(x,y)$ of pair of
strings as $\KS([x,y])$.

Note that $\KS(x,y)$ changes only by $O(1)$-term if we consider
another computable ``pairing function'': If $[ x,y]_1$ and
$[x,y]_2$ are two pairing functions, then $[x,y]_1$ can be
obtained from $[x,y]_2$ by an algorithm, so $\KS([x,y]_1)\le
\KS([x,y]_2)+O(1)$.

Note that
        $$
\KS(x,y)\ge \KS(x) \quad\text{and}\quad \KS(x,y) \ge \KS(y)
        $$
(indeed, there are computable functions that produce $x$ and $y$
from $[x,y]$).

For similar reasons, $\KS(x,y)=\KS(y,x)$ and $\KS(x,x)=\KS(x)$.

We can define $\KS(x,y,z)$, $\KS(x,y,z,t)$ etc. in a similar way:
$\KS(x,y,z)=\KS([x,[y,z]])$ (or $\KS(x,y,z)=\KS([[x,y],z])$, the
difference is $O(1)$).

\begin{theorem}
        $$
\KS(x,y)\le \KS(x)+2\log \KS(x)+ \KS(y) + O(1).
        $$
\end{theorem}

\begin{proof}

By $\overline x$ we denote binary string $x$ with all
bits doubled. Let $D$ be the optimal decompression algorithm.
Consider the following decompression algorithm $D_2$:
        $$
\overline{\bin(|p|)}01pq \mapsto [D(p), D(q)].
        $$
Note that $D_2$ is well defined, because the input string
$\overline{\bin(|p|)}01pq$ can be disassembled into parts
uniquely: we know where $01$ is, so we can find $|p|$ and then
separate $p$ and $q$.

If $p$ is the shortest description for $x$ and $q$ is the
shortest description for $y$, then $D(p)=x$, $D(q)=y$ and
$D_2(\overline{\bin(p)}01pq)=[x,y]$. Therefore
        $$
\KS_{D_2} ([x,y]) \le |p|+2\log|p|+|q|+O(1);
        $$
here $|p|=\KS(x)$ and $|q|=\KS(y)$ by our assumption.

\end{proof}

Of course, $p$ and $q$ can be exchanged: we can replace
$\log \KS(p)$  by $\log \KS(q)$.

\section{Conditional complexity}

We now want to define conditional complexity of $x$ when $y$ is
known. Imagine that you want to send string $x$ to your friend
using as few bits as possible. If she already knows some string
$y$ which is similar to $x$, this can be used to make the message shorter.

Here is the definition. Let $\langle p,y\rangle\mapsto D(p,y)$
be a computable function of two arguments. We define the conditional
complexity $\KS_D(x|y)$ of $x$ when $y$ is known as
        $$
\KS_D(x|y)=\min\{\,|p|\,\mid D(p,y)=x\}.
        $$
As usual, $\min(\varnothing)=+\infty$. The function $D$ is
called ``conditional decompressor'' or ``conditional
description mode'': $p$ is the description (compressed version)
of $x$ when $y$ is known. (To get $x$ from $p$ the decompressing
algorithm $D$ needs $y$.)

\begin{theorem}
There exists an optimal conditional decompressing function $D$
such that for every other conditional decompressing function $D'$
there exists a constant $c$ such that
        $$
\KS_D(x|y) \le \KS_{D'}(x|y) +c
        $$
for all strings $x$ and $y$.
\end{theorem}

\begin{proof}

As for the non-conditional version, consider some programming language where programs allow two input strings and are self-delimiting. Then let
        $$
D(uv,y)= \text{the output of program $u$ applied to $v,y$.}
        $$
Algorithm $D$ finds a (self-delimiting) program $u$ as a prefix
of its first argument and then applies $u$ to the rest of the
first argument and the second argument.

Let $D'$ be some other conditional decompressing function. Being
computable, it has some program~$u$. Then
        $$
\KS_D(x|y) \le \KS_{D'}(x|y) + |u|.
        $$
Indeed, let $p$ be the shortest string such that $D'(p,y)=x$
(therefore, $|p|=\KS_{D'}(x|y)$). Then $D(up,y)=x$, therefore
$\KS_D(x|y)\le |up|=|p|+|u|=\KS_{D'}(x|y)+|u|$.
\end{proof}

We fix some optimal conditional decompressing function $D$ and
omit the index $D$ in $\KS_D(x|y)$. Beware that $\KS(x|y)$ is defined
only ``up to $O(1)$-term''.

\begin{theorem}
\leavevmode\par
\textup{\textbf{(a)}}~$\KS(x|y) \le \KS(x)+O(1).$

\textup{\textbf{(b)}}~For every $y$ there exists some constant $c$ such that
        $$
 |\KS(x) - \KS(x|y)| \le c.
        $$
\end{theorem}

This theorem says that conditional complexity is smaller than the
unconditional one but for every fixed condition the difference is
bounded by a constant (depending on the condition).

\begin{proof}
(a)~If $D_0$ is an (unconditional) decompressing algorithm,
we can consider a conditional decompressing algorithm
        $$
D(p,y)=D_0(p)
        $$
that ignores conditions. Then $\KS_D(x|y)=\KS_{D_0}(x)$.

(b)~On the other hand, if $D$ is a conditional decompressing
algorithm, for every fixed $y$ we may consider an (unconditional)
decompressing algorithm $D_y$ defined as
        $$
D_y(p)=D(p,y).
        $$
Then $\KS_{D_y} (x)=\KS_D(x|y)$ for given $y$ and for all $x$. And
$\KS(x)\le \KS_{D_y}(x)+O(1)$ (where $O(1)$-constant depends on
$y$). 
\end{proof}

\section{Pair complexity and conditional complexity}

\begin{theorem}
        $$
\KS(x,y)=\KS(x|y)+\KS(y)+O(\log \KS(x)+\log \KS(y)).
        $$
\end{theorem}

\begin{proof}

Let us prove first that
        $$
\KS(x,y)\leq \KS(x|y)+\KS(y)+O(\log \KS(x)+\log \KS(y)).
        $$
We do it as before: If $D$ is an optimal decompressing function
(for unconditional complexity) and $D_2$ is an optimal
conditional decompressing function, let
        $$
D'(\overline{\bin(p)}01pq)=[D_2(p,D(q)),D(q)].
        $$
In other terms, to get the description of pair $x,y$ we
concatenate the shortest description of $y$ (denoted by $q$) with
the shortest description of $x$ when $y$ is known (denoted by
$p$). (Special precautions are used to guarantee the unique
decomposition.) Indeed, in this case $D(q)=y$ and
$D_2(p,D(q))=D_2(p,y)=x$, therefore
        \begin{multline*}
\KS_{D'}([x,y])\le |p|+2\log|p|+|q|+O(1) \le \\
      \le  \KS(x|y)+\KS(y)+O(\log \KS(x)+\log \KS(y)).
        \end{multline*}

The reverse inequality is much more interesting. Let us explain
the idea of the proof. This inequality is a translation of a
simple combinatorial statement. Let $A$ be a finite set of pairs
of strings. By $|A|$ we denote the cardinality of $A$. For each
string $y$ we consider the set $A_y$ defined as
        $$
A_y=\{ x | \langle x,y\rangle\in A\}.
        $$
The cardinality $|A_y|$ depends on $y$ (and is equal to~$0$ for
all $y$ outside some finite set). Evidently,
        $$
\sum_y |A_y| = |A|.
        $$
Therefore, the number of $y$ such that $|A_y|$ is big, is limited:
        $$
|\{ y |\, |A_y|\ge c\} | \le |A|/c
        $$
for each $c$.

Now we return to complexities. Let $x$ and $y$ be two strings.
The inequality $\KS(x|y)+\KS(y)\le \KS(x,y)+O(\log \KS(x)+\log \KS(y))$
can be informally read as follows: if $\KS(x,y)<m+n$, then either
$\KS(x|y)<m$ or $\KS(y)<n$ up to logarithmic terms. Why is it the
case? Consider a set $A$ of all pairs $\langle x,y\rangle$ such
that $\KS(x,y)<m+n$. There are at most $2^{m+n}$ pairs in $A$. The
given pair $\langle x,y\rangle$ belongs to~$A$. Consider the set
$A_y$. It is either ``small'' (contains at most $2^m$ elements)
or ``big'' (=not small). If $A_y$ is small ($|A_y|\le 2^m$), then $x$ can be
described (when $y$ is known) by its ordinal number in $A_y$,
which requires $m$ bits, and $\KS(x|y)$ does not exceed $m$ (plus
some administrative overhead). If $A_y$ is big, then $y$ belongs
to a (rather small) set $Y$ of all strings $y$ such that $A_y$
is big. The number of strings $y$ such that $|A_y|>2^m$ does not
exceed $|A|/2^m=2^n$. Therefore, $y$ can be (unconditionally)
described by its ordinal number in $Y$ which requires $n$ bits
(plus overhead of logarithmic size).

Let us repeat this more formally. Let $\KS(x,y)=a$. Consider the
set $A$ of all pairs $\langle x,y\rangle$ that have complexity
at most $a$. Let $b=\lfloor\log_2 |A_y|\rfloor$. To describe $x$
when $y$ is known we need to specify $a,b$ and the ordinal
number of $x$ in $A_y$ (this set can be enumerated effectively
if $a$ and $b$ are known since $\KS$ is enumerable from above).
This ordinal number has $b+O(1)$ bits and, therefore, $\KS(x|y)\le
b + O(\log a+\log b)$.

On the other hand, the set of all $y'$ such that $|A_{y'}|\ge
2^b$ consists of at most $|A|/2^b=O(2^{a-b})$ elements and can
be enumerated when $a$ and $b$ are known. Our $y$ belongs
to this set, therefore, $y$ can be described by $a$, $b$ and
$y$'s ordinal number, and $\KS(y)\le a-b + O(\log a+\log b)$.
Therefore, $\KS(y)+\KS(x|y)\le a + O(\log a +\log b)$. 

\end{proof}

\begin{problems}
\subsection*{Problems}

\leavevmode

1. Define $\KS(x,y,z)$ as $\KS([[x,y],[x,z]])$. Is this  definition
equivalent to a standard one (up to $O(1)$-term)?

2. Prove that $\KS(x,y)\le \KS(x)+\log K(x)+2\log\log \KS(x)+\KS(y)+O(1)$.
(Hint: repeat the trick with encoded length.)

3. Let $f$ be a computable function of two arguments. Prove that
$\KS(f(x,y)|y) \le \KS(x|y)+O(1)$ where $O(1)$-constant depends on
$f$ but not on $x$ and~$y$.

4$^{*}$. Prove that $\KS(x|\KS(x))=\KS(x)+O(1)$.
\end{problems}

\section{Applications of conditional complexity}

\begin{theorem}
If $x,y,z$ are strings of length at most $n$, then
        $$
2\KS(x,y,z) \le \KS(x,y)+\KS(x,z)+\KS(y,z)+O(\log n)
        $$
\end{theorem}

\begin{proof}

The statement does not mention conditional complexity;
however, the proof uses it. Recall that (up to $O(\log n)$-terms)
we have
        $$
\KS(x,y,z)-\KS(x,y)=\KS(z|x,y)
        $$
and
        $$
\KS(x,y,z)-\KS(x,z)=\KS(y|x,z)
        $$
Therefore, our inequality can be rewritten as
        $$
\KS(z|x,y)+\KS(y|x,z)\le \KS(y,z),
        $$
and the right-hand side is (up to $O(\log n)$) equal to
        $
\KS(z|y)+\KS(y).
        $
It remains to note that $\KS(z|x,y)\le \KS(z|y)$ (the more we know,
the smaller is the complexity) and $\KS(y|x,z)\le \KS(y)$.

\end{proof}

\section{Incompressible strings}

A string $x$ of length $n$ is called \emph{incompressible} if
$\KS(x|n)\ge n$. A more liberal definition: $x$ is $c$-incompressible,
if $\KS(x|n)\ge n-c$.

Note that this definition depends on the  choice of the optimal decompressor (but the difference can be covered by an $O(1)$-change in $c$).

\begin{theorem}
For each $n$ there exist incompressible strings of length $n$.
For each $n$ and each $c$ the fraction of $c$-incompressible
strings among all strings of length $n$ is greater than
$1-2^{-c}$.
\end{theorem}

\begin{proof}

The number of descriptions of length less than $n-c$ is
$1+2+4+\ldots+2^{n-c-1}<2^{n-c}$. Therefore, the fraction of
$c$-compressible strings is less than $2^{n-c}/2^n=2^{-c}$.

\end{proof}

\section{Computability and complexity of initial segments}

\begin{theorem}
        \label{computability-complexity}
An infinite sequence $x=x_1x_2x_3\dots$ of zeros and ones is
computable if and only if $\KS(x_1\ldots x_n|n)=O(1)$.
\end{theorem}

Proof. If $x$ is computable, then the initial segment $x_1\ldots x_n$
is a computable function of $n$, and $\KS(f(n)|n)=O(1)$ for every
computable function $f$.

The other direction is more complicated. We provide this
proof since it uses some methods that are typical for the general theory of computation (recursion theory).

Assume that $\KS(x_1\dots x_n|n)<c$ for some $c$ and all $n$.
We have to prove that the sequence $x_1x_2\ldots$ is computable.
Let us say that a string of length $n$ is ``simple'' if $\KS(x|n)<c$. There
are at most $2^c$ simple strings of each length. The set of
all simple strings is enumerable (we can generate them trying
all short descriptions in parallel for all $n$).

We call a string ``good'' if all its prefixes (including the
string itself) are simple. The set of all good strings is also
enumerable. (Enumerating simple strings, we can select strings
whose prefixes are found to be simple.)

Good strings form a subtree in full binary tree. (Full binary
tree is a set of all binary strings. A subset $T$ of full binary
tree is a subtree if all prefixes of every string $t\in T$ are
elements of $T$.)

The sequence $x_1x_2\ldots$ is an infinite branch of the subtree
of good strings. Note that this subtree has at most $2^c$ infinite
branches because each level has at most $2^c$ vertices.

Imagine for a while that subtree of good strings is decidable.
(In fact, it is not the case, and we will need additional
construction.) Then we can apply the following statement:

\textbf{Lemma 1}. \emph{If a decidable subtree has only finite number of infinite branches, all these branches are computable}.

\emph{Proof}. If two branches in a tree are different then they diverge
at some point and never meet again. Consider a level $N$ where
all infinite branches diverge. It is enough to show that for
each branch there is an algorithm that chooses the direction of
branch (left or right, i.e., $0$ or $1$) above level $N$. Since
we are above level $N$, the direction is determined uniquely: if
we choose a wrong direction, no infinite branches are possible.
By compactness (or K\"onig lemma), we know that in this case a subtree rooted in
the ``wrong'' vertex will be finite. This fact can be discovered
at some point (recall that subtree is assumed to be decidable).
Therefore, at each level we can wait until one of two possible
directions is closed, and choose another one. This algorithm works
only above level $N$, but the initial segment can be a
compiled-in constant. Lemma 1 is proven.

Application of Lemma 1 is made possible by the following statement:

\textbf{Lemma 2}. \emph{Let $G$ be a subtree of good strings. Then there exists a decidable subtree $G'\subset G$ that contains all infinite
branches of $G$}.

\emph{Proof}. For each $n$ let $g(n)$ be the number of good strings of
length $n$. Consider an integer $g=\lim\sup g(n)$. In other
words, there exist infinitely many $n$ such that $g(n)=g$
but only finitely many $n$ such that $g(n)>g$. We choose
some $N$ such that $g(n)\le g$ for all $n\ge N$ and
consider only levels $N,N+1,\dots$

A level $n\ge N$ is called \emph{complete} if $g(n)=g$. By our
assumption there are infinitely many complete levels. On the
other hand, the set of all complete levels is enumerable.
Therefore, we can construct a computable increasing sequence
$n_1<n_2<\ldots$ of complete levels. (To find $n_{i+1}$, we
enumerate complete levels until we find $n_{i+1}>n_{i}$.)

There is an algorithm that for each $i$ finds the list of all
good strings of length $n_i$. (It waits until $g$ goods strings
of length $n_i$ appear.) Let us call all those strings (for all
$i$) ``selected''. The set of all selected strings is decidable.
If a string of length $n_j$ is selected, then its prefix of
length $n_i$ (for $i<j$) is selected. It is easy to see now that
selected strings and their prefixes form a decidable subtree
$G'$ that includes all infinite branches of $G$.

Lemma 2 (and Theorem~\ref{computability-complexity})
are proven.

For a computable sequence $x_1x_2\dots$ we have
$\KS(x_1\ldots x_n|n)=O(1)$ and therefore $\KS(x_1\ldots x_n)\le\log
n+O(1)$. One can prove that this last (seemingly weaker) inequality also implies
computability of the sequence. However, the inequality $\KS(x_1\dots x_n)=O(\log n)$ does not imply computability of $x_1x_2\dots$, as the following result shows.

\begin{theorem}
Let $A$ be an enumerable set of natural numbers. Then for its
characteristic sequence $a_0a_1a_2\ldots$ \textup($a_i=1$ if $i\in A$
and $a_i=0$ otherwise\textup) we have
       $$
\KS(a_0a_1\ldots a_n)=O(\log n).
       $$
\end{theorem}

\begin{proof}

To specify $a_0\ldots a_n$ it is enough to specify two
numbers. The first is $n$ and the second is the number of $1$'s
in $a_0\ldots a_n$, i.e., the cardinality of the set $A\cap
[0,n]$. Indeed, for a given~$n$, we can enumerate this set, and
since we know its cardinality, we know when to stop the
enumeration. Both of them use $O(\log n)$ bits.
\end{proof}

This theorem shows that initial segments of characteristic
sequences of enumerable sets are far from being incompressible.

As we know that for each $n$ there exists an incompressible
sequence of length $n$, it is natural to ask whether there is an
infinite sequence $x_1x_2\dots$ such that its initial segment of
arbitrary length $n$ is incompressible (or at least $c$-incompressible
for some $c$ that does not depend on $n$). The following theorem
shows that it is not the case.

\begin{theorem}
There exists $c$ such that for every sequence $x_1x_2x_2\dots$
there are infinitely many $n$ such that
        $$
\KS(x_1x_2\ldots x_n)\le n-\log n + c
        $$
\end{theorem}

\begin{proof}

The main reason why it is the case is that the
series $\sum (1/n)$ diverges. It makes possible to
select the sets $A_1,A_2,\dots$ with following properties:

(1) each $A_i$ consists of strings of length $i$;

(2) $|A_i| \le 2^i/i$;

(3) for every infinite sequence $x_1x_2\ldots$ there are
   infinitely many $i$ such that\\ $x_1\ldots x_i\in A_i$.

(4) the set $A=\cup_i A_i$ is decidable.

Indeed, starting with some $A_i$, we cover about $(1/i)$-fraction
of the entire space $\Omega$ of all infinite sequences. Then we
can choose $A_{i+1}$ to cover other part of $\Omega$, and so on
until we cover all $\Omega$ (it happens because
$1/i+1/(i+1)+\ldots+1/j$ goes to infinity). Then we can start
again, providing a second layer of covering, etc.

It is easy to see that $|A_1|+|A_2|+\ldots+|A_i|=O(2^i/i)$:
Each term is almost twice as big as the preceding one, therefore,
the sum is $O(\text{last term})$.
Therefore, if we write down in lexicographic ordering
all the elements of $A_1,A_2,\ldots$, every element $x$ of $A_i$
will have ordinal number $O(2^i/i)$. This number determines $x$
uniquely and therefore for every $x\in A_i$ we have
        $$
\KS(x)\le \log (O(2^i)/i)=i-\log i + O(1).
        $$.
\end{proof}

\begin{problems}
\subsection*{Problems}

\leavevmode

1. True or false: for every computable function $f$ there exists a
constant $c$ such that $\KS(x|y)\le \KS(x|f(y))+c$ for all $x,y$
such that $f(y)$ is defined.

2. Prove that $\KS(x_1\ldots x_n|n)\le \log n + O(1)$ for every
characteristic sequence of an enumerable set.

3$^{*}$. Prove that there exists a sequence $x_1x_2\ldots$ such
that $\KS(x_1\ldots x_n)\ge n-2\log n - c$ for some $c$ and for
all $n$.

4$^{*}$. Prove that if $\KS(x_1\ldots x_n)\le\log n +c$ for some
$c$ and all $n$, then the sequence $x_1x_2\dots$ is computable.
\end{problems}

\section{Incompressibility and lower bounds}

In this section we show how to apply 
Kolmogorov complexity to obtain a lower bound for the
following problem. Let $M$ be a Turing machine (with one tape)
that duplicates its input: for every string $x$ on the tape (with
blanks on the right of $x$) it produces $xx$. We prove that $M$
requires time $\Omega(n^2)$ if $x$ is an incompressible string
of length $n$. The idea is simple: the head of TM can carry
finite number of bits with limited speed, therefore the speed of
information transfer (measured in bit$\times$cell$/$step) is
bounded and to move $n$ bits by $n$ cells we need $\Omega(n^2)$
steps.

\begin{theorem}
Let $M$ be a Turing machine. Then there exists some constant
$c$ with the following property: for every $k$, every $l\geq k$ and
every $t$, if cells $c_i$ with $i>k$ are initially empty, then the
complexity of the string $c_{l+1}c_{l+2}\ldots$ after $t$ steps
is bounded by $ct/(l-k)+O(\log l+\log t)$.
\end{theorem}

Roughly speaking, if we have to move information at least by
$l-k$ cells, then we can bring at most $ct/(l-k)$ bits into the
area where there was no information at the beginning.

One technical detail: string $c_{l+1}c_{l+2}\ldots$ denotes the
visited part of the tape (and is finite).

This theorem can be used to get a lower bound for duplication.
Let $x$ be an incompressible string of length $n$. We apply
duplicating machine to the string $0^nx$ (with $n$ zeros before
$x$). After the machine terminates in $t$ steps, the tape is
$0^nx0^nx$. Let $k=2n$ and $l=3n$. We can apply our theorem and
get $n\leqslant \KS(x)\leqslant ct/n+O(\log n+\log t)$. Therefore,
$t=\Omega(n^2)$ (note that $\log t<2\log n$ unless $t>n^2$).

\begin{proof}

Let $u$ be an arbitrary point on the tape between $k$ and $l$. A custom
officer records what TM carries is its head while crossing point
$u$ from left to right (but not the time of crossing). The
recorded sequence $T_u$ of TM-states is called \emph{trace} (at
point $u$). Each state occupies $O(1)$ bits since the set of
states is finite. This trace together with $u$, $k$, $l$ and the
number of steps after the last crossing (at most $t$) is enough
to reconstruct the contents of $c_{l+1}c_{l+2}\ldots$ at the
moment~$t$. (Indeed, we can simulate the behavior of $M$ on the
right of $u$.) Therefore, $\KS(c_{l+1}c_{l+2}\dots) \le cN_u
+O(\log l)+O(\log t)$ where $N_u$ is the length of $T_u$, i.e.,
the number of crossings at $u$.

Now we add these inequalities for all $u=k,k+1,\dots,l$. The
sum of $N_u$ is bounded by $t$ (since only one crossing is
possible at a given time). So
        $$
(l-k)K(c_{l+1}c_{l+2}\dots) \le t + (l-k)[O(\log l)+O(\log t)]
        $$
and our theorem is proven.

\end{proof}

The original result (one of the first lower bounds for time
complexity) was not for duplication but for palindrome
recognition: every TM that checks whether its input is a
palindrome (like \texttt{abadaba}) uses $\Omega(n^2)$ steps for
some inputs of length $n$.  This statement can also be proven by the incompressibility method.

Proof sketch: Consider a palindrome $xx^R$ of length $2n$. Let
$u$ be an arbitrary position in the first half of $xx^R$: $x=yz$ and
length of $y$ is $u$. Then the trace $T_u$ determines $y$
uniquely if we record states of TM while crossing checkpoint $u$
in both directions. Indeed, if strings with different $y$ have
the same trace, we can mix the left part of one computation with
the right part of another one and get a contradiction. Taking
all $u$ between $|x|/4$ and $|x|/2$, we get the required bound.

\section{Incompressibility and prime numbers}

Let us prove that there are infinitely many prime numbers.
Imagine that there are only $n$ prime numbers $p_1,\dots,p_n$.
Then each integer $N$ can be factored as
        $$
N=p_1^{k_1}p_2^{k_2}\ldots p_n^{k_n}.
        $$
where all $k_i$ do not exceed $\log N$. Therefore, each $N$ can
be described by $n$ integers $k_1,\dots,k_n$, and $k_i\leqslant
\log N$ for every $i$, so the total number of bits needed to
describe $N$ is $O(n\log\log N)$. But $N$ corresponds to
a string of length $\log N$, so we get a contradiction if this
string is incompressible.

\section{Incompressible matrices}

Consider an incompressible Boolean matrix of size $n\times n$.
Let us prove that its rank (over the field $\mathbb{F}_2=\{0,1\}$)
is greater than $n/2$.

Indeed, imagine that its rank is at most $n/2$. Then we can
select $n/2$ columns of the matrix such that all other columns are
linear combinations of the selected ones. Let $k_1,\ldots,k_{n/2}$
be the numbers of these columns.

Then, instead of specifying all bits of the matrix we can specify:

(1) the numbers $k_1,\dots,k_n$ ($O(n\log n)$ bits)

(2) bits in the selected columns ($n^2/2$ bits)

(3) $n^2/4$ bits that are coefficients in linear combinations of
    selected columns needed to get non-selected columns,
    ($n/2$ bits for each of $n/2$ non-selected columns).

Therefore, we get $0.75n^2+O(n\log n)$ bits instead of $n^2$
needed for incompressible matrix.

Of course, it is trivial to find a $n\times n$ Boolean matrix of
full rank, but this construction is interesting as an illustration of the incompressibility technique.

\section{Incompressible graphs}

An undirected graph with $n$ vertices can be represented by a bit string of length $n(n-1)/2$ (its adjacency matrix is symmetric). We call a graph \emph{incompressible} if this string is incompressible.

Let us show that an incompressible graph is necessarily connected. Indeed,
imagine that it can be divided into two connected components, and
one of them (the smaller one) has $k$ vertices ($k<n/2$). Then the graph
can be described by

(1) the list of numbers of $k$ vertices in this component ($k\log n$ bits), and

(2) $k(k-1)/2$ and $(n-k)(n-k-1)/2$ bits needed to describe
  both components.

In (2) (compared to the full description of the graph) we save
$k(n-k)$ bits for edges that go from one component to another
one, and $k(n-k)>O(k\log n)$ for big enough $n$ (recall that
$k<n/2$).

\section{Incompressible tournaments}

Let $M$ be a tournament, i.e., a complete directed graph with
$n$ vertices (for every two different vertices $i$ and $j$ there
exists either edge $i\to j$ or $j\to i$ but not both).

A tournament is \emph{transitive} if its vertices are linearly
ordered by the relation $i\to j$.
\medskip

\emph{Lemma}. \emph{Each tournament of size $2^k-1$ has a transitive
sub-tournament of size~$k$.}

\begin{proof}
 (Induction by $n$.) Let $x$ be a vertex. Then $2^k-2$
remaining vertices are divided into two groups: ``smaller'' than
$x$ and ``greater'' than $x$. At least one of the groups has
$2^{k-1}-1$ elements and contains a transitive sub-tournament of
size $k-1$. Adding $x$ to it, we get a transitive sub-tournament
of size $k$.
\end{proof}

This lemma gives a lower bound on the size of graph that does
not include transitive $k$-tournament.

The incompressibility method provides an upper bound: an
incompressible tournament with $n$ vertices may have transitive
sub-tournaments of $O(\log n)$ size only.

A tournament with $n$ vertices is represented by $n(n-1)/2$
bits. If a tournament $R$ with $n$ vertices has a transitive
sub-tournament $R'$  of size $k$, then $R$ can be described~by:

(1) the numbers of vertices in $R'$ listed according to linear
    $R'$-ordering ($k\log n$ bits), and

(2) remaining bits in the description of $R$ (except for bits
    that describe relations inside $R'$)

In (2) we save $k(k-1)/2$ bits, and in (1) we use $k\log n$
additional bits. Since we have to lose more than we win,
$k=O(\log n)$.

\section{Discussion}

All these results can be considered as direct reformulation of
counting (or probabilistic arguments). Moreover, counting gives
us better bounds without $O()$-notation.

But complexity arguments provide an important heuristics: We
want to prove that random object $x$ has some property and note
that if $x$ does not have it, then $x$ has some regularities
that can be used to give a short description for $x$.

\begin{problems}
\subsection*{Problems}

\leavevmode

1. Let $x$ be an incompressible string of length $n$ and let $y$
   be a longest substring of $x$ that contains only zeros. Prove
   that $|y|=O(\log n)$

2$^{*}$. Prove that $|y|=\Omega(\log n)$.

3. Let $w(n)$ be the largest integer such that for
   each tournament $T$ on $N=\{1,\dots,n\}$ there exist disjoint
   sets $A$ and $B$, each of cardinality $w(n)$, such that
   $A\times B\subseteq T$. Prove that $w(n)\le 2\lceil\log
   n\rceil$. (Hint: add $2w(n)\lceil\log n\rceil$ bit to
   describe nodes, and save $w(n)^2$ bits on edges. See~\cite{li-vitanyi} and~\cite{erdos-spencer}.)
\end{problems}

\section{$k$- and $k+1$-head automata}

A $k$-head finite automaton has $k$ (numbered)
heads that scan from left to right the input string (which is the same for all heads). Automaton has a finite number of states. Transition table specifies an action for each state and each $k$-tuple of input symbols. Action is a pair:
the new state, and the subset of heads to be moved. (We may assume that
at least one head should be moved; otherwise we can precompute
the next transition. We assume also that the input string is followed by blank symbols, so the automaton knows which heads have seen the entire input string.)

One of the states is called an \emph{initial} state. Some states
are \emph{accepting} states. An automaton $A$ accepts string $x$
if $A$ comes to an accepting state after reading $x$, starting
from the initial state and all heads placed at the left-most character. Reading $x$ is finished when all heads leave $x$. We require that this happens for arbitrary string $x$.

For $k=1$ we get the standard notion of finite automaton.

\emph{Example}: A $2$-head automaton can recognize strings of form
$x\#x$ (where $x$ is a binary string). The first head moves to
$\#$-symbol and then both heads move and check whether they see
the same symbols.

It is well known that this language cannot be recognized by
$1$-head finite automaton, so $2$-head automata are more
powerful that $1$-head ones.

Our goal is to prove the same separation between $k$-heads automata and $(k+1)$-heads automata for arbitrary~$k$.

\begin{theorem}
For every $k\ge 1$ there exists a language that can be recognized by
a $(k+1)$-head automaton but not by a $k$-head one.
\end{theorem}

\begin{proof}
The language is similar to the language considered above. For
example, for $k=2$ we consider a language consisting of strings
        $$
x\# y \# z \# z\# y \#x
        $$
Using three heads, we can easily recognize this language.
Indeed, the first head moves from left to right and ignores the
left part of the input string, while the second and the third one are moved to the left copies of $x$ and $y$. These copies are checked when the first head crosses the right copies of $y$ and $x$. Then only one unchecked string $z$ remains, and there are two heads at the left of it, so this can be done.  

The same approach shows that an automaton
with $k$ heads can recognize language $L_N$ that
consists of strings
        $$
x_1 \# x_2\# \dots\# x_N \# x_N \#\dots\# x_2\#x_1
        $$
for $N=(k-1)+(k-2)+\ldots+1= k(k-1)/2$ (and for all smaller $N$).

Let us prove now that $k$-head automaton $A$ cannot recognize
$L_N$ if $N$ is bigger than $k(k-1)/2$. (In particular, no
automaton with $2$ heads can recognize $L_3$ and even $L_2$.)

Let us fix a string
        $$
x=x_1 \# x_2\# \dots\# x_N \# x_N \#\dots\# x_2\#x_1
        $$
where all $x_i$ have the same length $l$ and the string
$x_1x_2\ldots x_N$ is an incompressible string (of length $Nl$).
String $x$ is accepted by $A$. In our argument the following
notion is crucial: We say that an (unordered) pair of heads
``covers'' $x_m$ if at some point one head is inside the left
copy of $x_m$ while the other head (from this pair) is inside the right
copy.

After that the right head can visit only strings
$x_{m-1},\dots,x_1$ and left head cannot visit the left counterparts
of those strings (they are on the left of it). Therefore, only
one $x_m$ can be covered by a given pair of heads.

In our example we had three heads (and, therefore, three pairs of
heads) and each string $x_1,x_2,x_3$ was covered by one pair.

The number of pairs is $k(k-1)/2$ for $k$ heads. Therefore (since $N>k(k-1)/2$) there exists some $x_m$ that was not covered at all during the
computation. We show that conditional complexity of $x_m$ when
all other $x_i$ are known does not exceed $O(\log l)$. (The
constant here depends on $N$ and $A$, but not on $l$.) This
contradicts to the incompressibility of $x_1\ldots x_N$ (we can
replace $x_m$ by self-delimiting description of $x_m$ when other
$x_i$ are known and get a shorter description of an incompressible
string).

The bound for the conditional complexity of $x_m$ can be
obtained in the following way. During the accepting computation
we take special care of the periods when one of the heads is
inside $x_m$ (on the left or on the right). We call these periods ``critical
sections''. Note that each critical section is either L-critical
(some heads are inside the left copy of $x_m$) or R-critical but not
both (no pair of heads covers $x_m$). Critical section starts
when one of the heads moves inside $x_m$ (other heads can also
move in during the section) and ends when all heads leave $x_m$.
Therefore, the number of critical sections during the computation
is at most $2k$.

Let us record the positions of all heads and the state of
automaton at the beginning and at the end of each critical
section. This requires $O(\log l)$ bits (note that we do not
record time).

We claim that this information (called \emph{trace} in the
sequel) determines $x_m$ if all other $x_i$ are known. To see
why, let us consider two computations with different $x_m$ and
$x_m'$ but the same $x_i$ for $i\ne m$ and the same traces.

Equal traces allow us to ``cut and paste'' these two
computations on the boundaries of critical sections. (Outside
the critical sections computations are the same, because the
strings are identical except for $x_m$, and state and positions
after each critical section are included in a trace.) Now we
take L-critical sections from one computation and R-critical
sections from another one. We get a mixed computation that is
an accepting run of $A$ on a string that has $x_m$ on the
left and $x_m'$ on the right. Therefore, $A$ accepts a string that
it should not accept. 
\end{proof}

\section{Heap sort: time analysis}

Let us assume that we sort numbers $1,2,\dots,N$. We have $N!$
possible permutations. Therefore, to specify a permutation we
need about $\log (N!)$ bits. Stirling's formula says that $N!\approx
(N/e)^N$, therefore the number of bits needed to specify one
permutation is $N\log N + O(N)$. As usual, most of the
permutations are incompressible in the sense that they have
complexity at least $N\log N-O(N)$. We estimate the number of
operations for heap sort in the case of an incompressible permutation.

Heap sort (we assume in this section that the reader knows what it is) consists of two phases. First phase creates a heap out
of the input array. (The indexes in array $a[1..N]$ form a tree where $2i$
and $2i+1$ are sons of $i$. The heap property says that ancestor has
bigger value that its descendants.)

Transforming the array into a heap goes as follows: for each
$i=N,N-1,\dots,1$ we make the heap out of subtree rooted at $i$ assuming that $j$-subtrees for $j>i$ are heaps.
Doing this for the node $i$, we need $O(k)$ steps where $k$ is the
distance between node $i$ and the leaves of the tree. Here
$k=0$ for about half of the nodes, $k=1$ for about $1/4$ of the nodes
etc., and the average number of steps per node is $O(\sum
k2^{-k})=O(1)$; the total number of operations is $O(N)$.

Important observation: after the heap is created, the complexity
of array $a[1..N]$ is still $N\log N+O(N)$, if the initial
permutation was incompressible. Indeed, ``heapifying'' means
composing the initial permutation with some other permutation
(which is determined by results of comparisons between array
elements). Since the total time for heapifying is $O(N)$, there are
at most $O(N)$ comparisons and their results form a bit string
of length $O(N)$ that determines the heapifying permutation. The
initial (incompressible) permutation is a composition of the
heap and $O(N)$-permutation, therefore heap has complexity at
least $N\log N-O(N)$.

The second phase transforms the heap into a sorted array. At every stage the
array is divided into two parts: $a[1..n]$ is still a heap, but
$a[n+1..N]$ is the end of the sorted array. One step of
transformation (it decreases $n$ by $1$) goes as follows: the
maximal heap element $a[1]$ is taken out of the heap and
exchanged with $a[n]$. Therefore, $a[n..N]$ is now sorted, and the
heap property is almost true: ascendant has bigger value
that descendant unless ascendant is $a[n]$ (that is now in root
position). To restore heap property, we move $a[n]$ down the
heap. The question is how many steps do we need. If the final
position is $d_n$ levels above the leaves level, we need $\log
N-d_n$ exchanges, and the total number of exchanges is $N\log
N-\sum d_n$.

We claim that $\sum d_n=O(N)$ for incompressible permutations,
and, therefore, the total number of exchanges is $N\log N+O(N)$.

So why $\sum d_n$ is $O(N)$? Let us record the direction of
movements while elements fall down through the heap (using
$0$ and $1$ for left and right). We don't use delimiters to
separate strings that correspond to different $n$ and use $N\log
N-\sum d_i$ bits altogether. Separately we write down all $d_n$
in self-delimiting way. This requires $\sum (2\log d_i+O(1))$ bits.
All this information allows us to reconstruct the exchanges during
the second phase, and therefore to reconstruct the initial state
of the heap before the second phase. Therefore, the complexity
of heap before the second phase (which is $N\log N-O(N)$) does not
exceed $N\log N -\sum d_n  + \sum(2\log d_n) + O(N)$, therefore,
$\sum (d_n-2\log d_n)=O(N)$. Since $2\log d_n < 0.5d_n$ for
$d_n>16$ (and all smaller $d_n$ have sum $O(N)$ anyway),
we conclude that $\sum d_n=O(N)$.

\begin{problems}
\subsection*{Problems}

\leavevmode

1$^{*}$. Prove that for most pairs of binary strings $x,y$ of
length $n$ every common subsequence of $x$ and $y$ has length at
most $0.99n$ (for large enough $n$).
\end{problems}

\section{Infinite random sequences}

There is some intuitive feeling saying that
a fair coin tossing cannot produce sequence
        $$
00000000000000000000000\dots
        $$
or
        $$
01010101010101010101010\dots,
        $$
so infinite sequences of zeros and ones can be divided in
two categories. \emph{Random} sequences are sequences that are plausible outcomes of  coin tossing; \emph{non-random} sequences
(including the two sequences above) are not plausible. It is more difficult
to provide an example of a random sequence (it somehow becomes
non-random after the example is provided), so our intuition is
not very reliable here.

\section{Classical probability theory}

Let $\Omega$ be the set of all infinite sequences of zeros and
ones. We define the \emph{uniform Bernoulli measure} on $\Omega$
as follows. For each binary string $x$ let $\Omega_x$ be the set
of all sequences that have prefix $x$ (a subtree rooted at $x$).

Consider a measure $P$ such that $P(\Omega_x)=2^{-|x|}$.
Measure theory allows us to extend this measure to all Borel
sets (and even further).

A set $X\subset \Omega$ is called a \emph{null} set if $P(X)$
is defined and $P(X)=0$. Let us give a direct equivalent
definition that is useful for constructive version:

A set $X\subset \Omega$ is a null set if for every $\varepsilon>0$
there exists a sequence of binary strings $x_0,x_1,\dots$ such that

(1) $X \subset \Omega_{x_0}\cup\Omega_{x_1}\cup\ldots$;

(2) $\sum\limits_i 2^{-|x_i|} < \varepsilon$.

Note that $2^{-|x_i|}$ is $P(\Omega_{x_i})$ according to our
definition. In words: $X$ is a null set if it can be covered by
a sequence of intervals $\Omega_{x_i}$ of arbitrarily small total measure.

\emph{Examples}: Each singleton is a null set. A countable union of
null sets is a null set. A subset of a null set is a null set.
The set $\Omega$ is not a null set (by compactness). The set of all
sequences that have zeros at positions with even numbers is a
null set.

\section{Strong Law of Large Numbers}

Informally, the strong law of large numbers (SLLN) says that random sequences $x_0x_1\ldots$
have limit frequency $1/2$, i.e.,
        $$
\lim_{n\to\infty} \frac{x_0+x_1+\ldots+x_{n-1}}{n}=\frac{1}{2}.
        $$
However, the word ``random'' here is used only as a shortcut:
the full meaning is that the set of all sequences that do not
satisfy SLLN (do not have limit
frequency or have it different from $1/2$) is a null set.

In general, when people say that``$P(\omega)$ is true for random $\omega\in\Omega$'', it usually means that the set
        $$
\{\omega\mid \text{$P(\omega)$ is false}\}
        $$
is a null set.

\emph{Proof sketch for SLLN}: it is enough to show that for every
$\delta>0$ the set $N_{\delta}$ of sequences that have
frequency greater than $1/2+\delta$ for infinitely many
prefixes, has measure $0$. (After that we use that a countable
union of null sets is a null set.) For each $n$ consider the
probability $p(n,\delta)$ of the event ``random string of
length $n$ has more than $(1/2+\delta)n$ ones''. The
crucial observation is that
        $$
\sum_{n} p(n,\delta) < \infty
        $$
for each $\delta>0$. (Actually, $p(n,\delta)$ is
exponentially decreasing as $n\to\infty$; proof uses
Stirling's approximation for factorials.) If the series
above has a finite sum, for every $\varepsilon >0$ one can find
an integer $N$ such that
        $$
\sum_{n>N} p(n,\delta) < \epsilon.
        $$
Consider all strings $z$ of length greater than $N$ that have
frequency of ones greater than $1/2+\delta$. The sum of
$P(\Omega_z)$ is equal to $\sum_{n>N} p(n,\delta) < \epsilon,$
and $N_{\varepsilon}$ is covered by family $\Omega_z$. 

\section{Effectively null sets}

The following notion was introduced by Per Martin-L\"of. A set
$X\subset \Omega$ is an \emph{effectively null} set if there is
an algorithm that gets a rational number $\varepsilon>0$ as
input and enumerates a set of strings
$\{x_0,x_1,x_2,\dots\}$ such that

(1) $X \subset \Omega_{x_0}\cup\Omega_{x_1}\cup\Omega_{x_2}\cup\ldots$;

(2) $\sum\limits_i 2^{-|x_i|} < \varepsilon$.

The notion of effectively null set remains the same if we allow
only $\varepsilon$ of form $1/2^k$, or if we replace ``$<$'' by
``$\le$'' in (2).

Every subset of an effectively null set is also an effectively
null set (evident observation).

For a computable infinite sequence $\omega$ of zeros and ones the singleton $\{\omega\}$ is a null set. (The same happens for all non-random $\omega$, see below.) 

An union of two effectively null sets is an effectively null set.
(Indeed, we can find enumerable coverings of size
$\varepsilon/2$ for both and combine them.)

More general statement requires preliminary definition. By
``covering algorithm'' for an effectively null set we mean an
algorithm mentioned in the definition (that gets $\varepsilon$
and generates a covering sequence of strings with sum of
measures less than~$\varepsilon$).

\textbf{Lemma}. \emph{Let $X_0,X_1,X_2,\dots$ be a sequence of effectively null
sets such that there exists an algorithm that given an integer $i$
produces \textup(some\textup) covering algorithm for $X_i$. Then $\cup X_i$ is
an effectively null set}.

\begin{proof}
To get an $\varepsilon$-covering for $\cup X_i$, we put
together $(\varepsilon/2)$-covering for $X_0$,
$(\varepsilon/4)$-covering for $X_1$, etc. To generate this
combined covering, we use the algorithm that produces covering for
$X_i$ from $i$. 
\end{proof}

\section{Maximal effectively null set}

Up to now the theory of effectively null sets just repeats the
classical theory of null sets. The crucial difference is in the
following theorem (proved by Martin-L\"of):

\begin{theorem}
        \label{maximal-null}
There exists a maximal effectively null set, i.e., an
effectively null set $N$ such that $X\subset N$ for every
effectively null set $X$.
\end{theorem}

(Trivial) reformulation: the union of all effectively null sets
is an effectively null set.

\begin{proof}
We cannot prove this theorem by applying the above lemma to all
effectively null sets (there are uncountably many of them, since every subset of an effectively null set is an effectively null
set).

But we don't need to consider all effectively null sets; it is
enough to consider all covering algorithms. For a given
algorithm (that gets positive rational number as input and
generates binary strings) we cannot say (effectively) whether it
is a covering algorithm or not. But we may artificially enforce
some restrictions: if algorithm (for given $\varepsilon>0$)
generates strings $x_0,x_1,\ldots$, we can check whether
$2^{-|x_0|}+\ldots+2^{-|x_k|}<\varepsilon$ or not; if not, we
delete $x_k$ from the generated sequence. Let us denote by $A'$
the modified algorithm (if $A$ was an original one). It is easy
to see that

(1) if $A$ was a covering algorithm for some effectively null
set, then $A'$ is equivalent to $A$ (the condition that we
enforce is never violated).

(2) For every $A$ the algorithm $A'$ is (almost) a covering algorithm
for some null set; the only difference is that the infinite sum
$\sum 2^{-|x_i|}$ can be equal to $\varepsilon$ even if all
finite sums are strictly less than $\varepsilon$.

But this is not important: we can apply the same arguments (that
were used to prove Lemma) to all algorithms $A'_0,A'_1,\ldots$
where $A_0,A_1,\ldots$ is a sequence of all algorithms (that get
positive rational numbers as inputs and enumerate sets of binary
strings). 

\end{proof}

\textbf{Definition}.
A sequence $\omega$ of zeros and ones is
called (Martin-L\"of) \emph{random} with respect to the uniform
Bernoulli measure if $\omega$ does not belong to the maximal
effectively null set.

(Reformulation: ``\dots if $\omega$ does not belong to any
effectively null set.'' )
\medskip

Therefore, to prove that some sequence is non-random we need to show that it belongs to some effectively null set.

Note also that a set $X$ is an effectively null set if and only
if all elements of $X$ are non-random.

This sounds like a paradox for people familiar with classical
measure theory. Indeed, we know that measure somehow reflects
the ``size'' of set. Each point is a null set, but if we have
too many points, we get a non-null set. Here (in Martin-L\"of
theory) the situation is different: if each element of some set forms an effectively null singleton (i.e., is non-random), then the entire set is an
effectively null one.

\begin{problems}
\subsection*{Problems}

\leavevmode

1. Prove that if sequence $x_0x_1x_2\dots$ of zeros and ones is
(Martin-L\"of) random with respect to uniform Bernoulli
measure, then the sequence $000x_1x_2\dots$ is also random.
Moreover, adding arbitrary finite prefix to a random sequence, we get a
random sequence, and adding arbitrary finite prefix to a non-random
sequence, we get a non-random sequence.

2. Prove that every (finite) binary string appears infinitely many
times in every random sequence.

3. Prove that every computable sequence is non-random. Give an
example of a non-computable non-random sequence.

4. Prove that the set of all computable infinite sequences of
zeros and ones is an effectively null set.

5$^*$. Prove that if a sequence $x_0x_1\dots$ is not random, then 
$n-\KS(x_0\dots x_{n-1}|n)$ tends to infinity as $n\to\infty$.
\end{problems}

\section{Gambling and selection rules}

Richard von Mises suggested (around 1910) the following notion
of a random sequence (he uses German word \emph{Kollektiv}) as a
basis for probability theory. A sequence $x_0x_1x_2\dots$
is called (Mises) random, if

(1)~it satisfies the strong law of large numbers, i.e., the limit frequency of $1$'s  in it is $1/2$:
        $$
\lim_{n\to\infty} \frac{x_0+x_1+\dots+x_{n-1}}{n} = \frac{1}{2};
        $$

(2) the same is true for every infinite subsequence selected by an
``admissible selection rule''.

Examples of admissible selection rules: (a)~select terms with even
indices; (b)~select terms that follow zeros. The first
rule gives $0100\dots$ when applied to
$\underline{0}0\underline{1}0\underline{0}1\underline{0}0\dots$
(selected terms are underlined). The second rule gives
$0110\dots$ when applied to
$0\underline{0}\underline{1}0\underline{1}10\underline{0}\dots$

Mises gave no exact definition of admissible selection rule
(at that time the theory of algorithms did not exist yet).
Later Church suggested the following formal definition
of admissible selection rule.

An admissible selection rule is a total computable function $S$
defined on finite strings that has values $1$ (``select'') and
$0$ (``do not select''). To apply $S$ to a sequence
$x_0x_1x_2\dots$ we select all $x_n$ such that
$S(x_0x_1\dots x_{n-1})=1$. Selected terms form a subsequence
(finite or infinite). Therefore, each selection rule $S$ determines
a mapping $\sigma_S:\Omega\to\Sigma$, where $\Sigma$ is the
set of all finite and infinite sequences of zeros and ones.

For example, if $S(x)=1$ for every string $x$, then $\sigma_S$ is
an identity mapping. Therefore, the first requirement in Mises
approach follows from the second one, and we come to the
following definition:

A sequence $x=x_0x_1x_2\dots$ is \emph{Mises--Church random}, if
for every admissible selection rule $S$ the sequence $\sigma_S(x)$
is either finite or has limit frequency~$1/2$.

Church's definition of admissible selection rules has the
following motivation. Imagine you come to a casino and watch the
outcomes of coin tossing. Then you decide whether to participate
in the next game or not, applying $S$ to the sequence of
observed outcomes.

\section{Selection rules and Martin-L\"of randomness}

\begin{theorem}
        \label{mises-ml}
Applying an admissible selection rule \textup(according to Church
definition\textup) to a Martin-L\"of random sequence, we get either a
finite sequence or a Martin-L\"of random sequence.
\end{theorem}

\begin{proof}
Let $S$ be a function that determines selection rule
$\sigma_S$.

Let $\Sigma_x$ be the set of all finite of infinite sequences
that have prefix $x$ (here $x$ is a finite binary string).

Consider the set $A_x=\sigma_S^{-1}(\Sigma_x)$ of all (infinite)
sequences $\omega$ such that selected subsequence starts with
$x$. If $x=\Lambda$ (empty string), then $A_x=\Omega$.

\textbf{Lemma}. \emph{The set $A_x$ has measure at most $2^{-|x|}$}.

\begin{proof}
What is $A_0$? In other terms, what is the set of all sequences
$\omega$ such that the selected subsequence (according to
selection rule $\sigma_S$) starts with $0$? Consider the set $B$
of all strings $z$ such that $S(z)=1$ but $S(z')=0$ for each 
prefix $z'$ of $z$. These strings mark the places where the first
bet is made. Therefore,
        $$
A_0= \cup\{\Omega_{z0} \mid z\in B\}
        $$
and
        $$
A_1= \cup\{\Omega_{z1} \mid z\in B\}.
        $$
In particular, the sets $A_0$ and $A_1$ have the same measure
and are disjoint, therefore
        $$
P(A_0)= P(A_1) \le \frac{1}{2}.
        $$
From the probability theory viewpoint, $P(A_0)$ [resp., $P(A_1)$] is
the probability of the event ``the first selected term will be
$0$ [resp. $1$]'', and both events have the same probability
(that does not exceed $1/2$) for evident reasons.

We can prove in the same way that $A_{00}$ and $A_{01}$ have the
same measure. (See below the details.) Since they are disjoint
subsets of $A_0$, both of them have measure at most $1/4$. The sets
$A_{10}$ and $A_{11}$ also have equal measure and
are subsets of $A_1$, therefore both have measure at most $1/4$,
etc.

If this does not sound convincing, let us give an explicit description of $A_{00}$. Let $B_0$ be the set of all strings $z$ such that

(1) $S(z)=1$;

(2) there exists exactly one proper prefix $z'$ of $z$ such that
     $S(z')=1$;

(3) $z'0$ is a prefix of $z$.

In other terms, $B_0$ corresponds to the positions where we are
making our second bet while our first bet produced $0$. Then
        $$
A_{00}=\cup\{\Omega_{z0} \mid z\in B_0\}
        $$
and
        $$
A_{01}=\cup\{\Omega_{z1} \mid z\in B_0\}.
        $$
Therefore $A_{00}$ and $A_{01}$ indeed have equal measures.

Lemma is proven.

\end{proof}

It is also clear that $A_x$ is the union of intervals $\Sigma_y$
that can be effectively generated if $x$ is known. (Here we use
the computability of $S$.)

Proving Theorem~\ref{mises-ml}, assume that $\sigma_S(\omega)$ is an infinite non-random
sequence. Then $\{\omega\}$ is effectively null singleton. Therefore,
for each $\varepsilon$ one can effectively generate intervals
$\Omega_{x_1},\Omega_{x_2},\dots$ whose union covers
$\sigma_S(\omega)$. The preimages $$\sigma_S^{-1}(\Sigma_{x_1}),
\sigma_S^{-1}(\Sigma_{x_2}),\dots$$ cover $\omega$. Each of these preimages
is an enumerable union of intervals, and if we combine all these
intervals we get a covering for $\omega$ that has measure less
than $\varepsilon$. Thus, $\omega$ is non-random, so Theorem~\ref{mises-ml} is proven.
\end{proof}

\begin{theorem}
Every Martin-L\"of random sequence has limit frequency $1/2$.
\end{theorem}

\begin{proof}
By definition this means that the set $\lnot SLLN$ of all
sequences that do not satisfy SLLN is an
effectively null set. As we have mentioned, this is a null set
and the proof relies on an upper bound for binomial
coefficients. This upper bound is explicit, and the argument
showing that the set $\lnot SLLN$ is a null set can be extended
to show that $\lnot SLLN$ is an effectively null set. 
\end{proof}

Combining these two results, we get the following

\begin{theorem}
Every Martin-L\"of random sequence is also Mises--Church random.
\end{theorem}

\begin{problems}
\subsection*{Problems}

\leavevmode

1. The following selection rule is \emph{not} admissible
according to Mises definition: choose all terms $x_{2n}$ such
that $x_{2n+1}=0$. Show that (nevertheless) it gives
(Martin-L\"of) random sequence if applied to a Martin-L\"of
random sequence.

2. Let $x_0x_1x_2\dots$ be a Mises--Church random sequence. Let
$a_N=|\{ n < N \mid x_n=0,\ x_{n+1}=1\}|$. Prove that
$a_N/N\to1/4$ as $N\to\infty$.
\end{problems}

\section{Probabilistic machines}

Consider a Turing machine that has access to a source of random
bits. Imagine, for example, that it has some special states $a,b,c$ with the following
properties: when the machine reaches state $a$, it jumps at the next step to one of the
states $b$ and $c$ with probability $1/2$ for each.

Another approach: consider a program in some language that allows assignments
        $$
a:=\texttt{random};
        $$
where $\texttt{random}$ is a keyword and $a$ is a Boolean variable that
gets value $0$ or $1$ when this statement is executed (with probability $1/2$; each new random bit is independent of the previous ones).

For a deterministic machine output is a function of its input. Now
it is not the case: for a given input machine can produce
different outputs, and each output has some probability. So for each input the output is a random variable. What can be said about this variable? We will consider machines without inputs; each machine of this type determines a random variable (its output).

Let $M$ be a machine without input. (For example, $M$ can be
a Turing machine that is put to work on an empty tape, or a Pascal
program that does not have \texttt{read} statements.) Now consider
probability of the event ``$M$ terminates''. What can be said about
this number?

More formally, for each sequence $\omega\in\Omega$ we consider
the behavior of $M$ if random bits are taken from $\omega$. For
a given $\omega$ the machine either terminates or not. Then $p$
is the measure of the set $T$ of all $\omega$ such that $M$
terminates using $\omega$. It is easy to see that $T$ is
measurable. Indeed, $T$ is a union of $T_n$, where $T_n$ is the
set of all $\omega$ such that $M$ stops after at most $n$ steps
using $\omega$. Each $T_n$ is a union of intervals $\Omega_t$
for some strings $t$ of length at most $n$ (machine can use at
most $n$ random bits if it runs in time $n$) and therefore is
measurable; the union of all $T_n$ is an open (and therefore measurable) set.

A real number $p$ is called \emph{enumerable from below} or
\emph{lower semicomputable} if $p$ is a limit of increasing
computable sequence of rational numbers: $p=\lim p_i$, where
$p_0 \le p_1\le p_2\le\dots$ and there is an algorithm that
computes $p_i$ given $i$.
\smallskip

\textbf{Lemma}. \emph{A real number $p$ is lower semicomputable if and only if the set $X_p=\{r\in \mathbb{Q} \mid r < p\}$ is \textup(computably\textup) enumerable}.

\begin{proof}
(1)~Let $p$ be the limit of a computable increasing sequence $p_i$.
For every rational number $r$ we have
        $$
r < p \Leftrightarrow \exists i\, [r < p_i].
        $$
Let $r_0,r_1,\dots$ be a computable sequence of rational numbers
such that every rational number appears infinitely often in this
sequence. The following algorithm enumerates $X_p$: at $i$th
step, compare $r_i$ and $p_i$; if $r_i<p_i$, output $r_i$.

(2)~If $X_p$ is computably enumerable, let $r_0,r_1,r_2,\dots$ be its
enumeration. Then $p_n=\max(r_0,r_1,\dots,r_n)$ is a non-decreasing
computable sequence of rational numbers that converges to $p$.
\end{proof}

\begin{theorem}
\textbf{\textup{(a)}}~Let $M$ be a probabilistic machine without input. Then $M$'s
probability of termination is lower semicomputable.

\textbf{\textup{(b)}}~Let $p$ be a lower semicomputable number in $[0,1]$.
Then there exists a probabilistic machine that terminates with
probability $p$.
\end{theorem}

\begin{proof}
(a)~Let $M$ be a probabilistic machine. Let $p_n$ be the
probability that $M$ terminates after at most $n$ steps. The
number $p_n$ is a rational number with denominator $2^n$ that
can be effectively computed for a given $n$. (Indeed, the machine
$M$ can use at most $n$ random bits during $n$ steps. For each
of $2^n$ binary strings we simulate behavior of $M$ and see for
how many of them $M$ terminates.) The sequence
$p_0,p_1,p_2\dots$ is an increasing computable sequence of
rational numbers that converges to~$p$.

(b)~Let $p$ be a real number in $[0,1]$ that is lower semicomputable.
Let $p_0\le p_1\le p_2\le\ldots$ be an increasing computable
sequence that converges to $p$. Consider the following
probabilistic machine. It treats random bits $b_0,b_1,b_2\dots$
as binary digits of a real number
        $$
\beta=0.b_0b_1b_2\ldots
        $$
When $i$ random bits are generated, we have lower and upper
bounds for $\beta$ that differ by $2^{-i}$. If the upper bound
$\beta_i$ turns out to be less than $p_i$, machine terminates.
It is easy to see that machine terminates for given
$\beta=0.b_0b_1\dots$ if and only if $\beta<p$. Indeed, if an upper
bound for $\beta$ is less than a lower bound for $p$, then $\beta
<p$. On the other hand, if $\beta <p$, then $\beta_i <p_i$ for
some $i$ (since $\beta_i\to\beta$ and $p_i\to p$ as
$i\to\infty$). 

\end{proof}

Now we consider probabilities of different outputs. Here we need
the following definition: A sequence $p_0,p_1,p_2\dots$ of real numbers
is \emph{lower semicomputable}, if there is a computable
total function~$p$ of two variables (that range over natural
numbers) with rational values (with special value
$-\infty$ added) such that
        $$
p(i,0)\le p(i,1)\le p(i,2)\le \ldots
        $$
and
        $$
p(i,0), p(i,1), p(i,2),\ldots \to p_i
        $$
for every~$i$.
\smallskip

\textbf{Lemma}. \emph{A sequence $p_0,p_1,p_2,\dots$ of reals is lower semicomputable if and only if the set of pairs
        $$
\{\langle i,r\rangle \mid r < p_i\}
        $$
is enumerable}.

\begin{proof}

Let $p_0,p_1,\dots$ be lower semicomputable and
$p_i=\lim_n p(i,n)$. Then
        $$
r < p_i \Leftrightarrow \exists n\, [r < p(i,n)]
        $$
and we can check $r<p(i,n)$ for all pairs $\langle i,r\rangle$
and for all $n$. If $r<p(i,n)$, pair $\langle i,r\rangle$ is
included in the enumeration.

On the other hand, if the set of pairs is enumerable, for each
$n$ we let $p(i,n)$ be the maximum value of $r$ for all pairs
$\langle i,r\rangle$ (with given $i$) that appear during $n$
steps of the enumeration process. (If there are no pairs,
$p(i,n)=-\infty$.) The lemma is proven.

\end{proof}

\begin{theorem}
\textup{(}a\textup{)}~Let $M$ be a probabilistic machine without
input that can produce natural numbers as outputs. Let $p_i$ be
the probability of the event ``$M$ terminates with output $i$''.
Then sequence $p_0,p_1,\dots$ is lower semicomputable and
$\sum_i p_i\le 1$.

\textup{(}b\textup{)}~Let $p_0,p_1,p_2\dots$ be a sequence of
non-negative real numbers that is lower semicomputable, and
$\sum_i p_i \le 1$. Then there exists a probabilistic machine
$M$ that outputs $i$ with probability \textup{(}exactly\textup{)}
$p_i$.
\end{theorem}

\begin{proof}
Part~(a) is similar to the previous argument: let
$p(i,n)$ be the probability that $M$ terminates with output $i$
after at most $n$ steps. Than $p(i,0),p(i,1),\dots$ is a
computable sequence of increasing rational numbers that
converges to $p_i$.

(b)~is more complicated. Recall the proof of the previous
theorem. There we had a ``random real'' $\beta$ and
``termination region'' $[0,p)$ where $p$ was the desired
termination probability. (If $\beta$ is in termination region,
machine terminates.)

Now termination region is divided into parts. For each
output value~$i$ there is a part of termination region that
corresponds to $i$ and has measure $p_i$. Machines terminates
with output~$i$ if and only if $\beta$ is inside $i$th part.

Let us consider first a special case when sequence $p_i$ is a
computable sequence of rational numbers, Then $i$th part is a segment
of length $p_i$. These segments are allocated from left to right
according to ``requests'' $p_i$. One can say that each number
$i$ comes with request $p_i$ for space allocation, and this
request is granted. Since we can compute the endpoints of all
segments, and have lower and upper bound for $\beta$, we are
able to detect the moment when $\beta$ is guaranteed to be inside
$i$-th part. 

In the general case the construction should be modified. Now each $i$
comes to space allocator many times with increasing requests
$p(i,0), p(i,1), p(i,2),\ldots$; each time the request is granted
by allocating additional interval of length $p(i,n)-p(i,n-1)$.
Note that now $i$th part is not contiguous: it consists of
infinitely many segments separated by other parts. But this is not important. Machine terminates with output $i$ when
current lower and upper bounds for $\beta$ guarantee that $\beta$
is inside $i$th part. The interior of $i$th part is a
countable union of intervals, and if $\beta$ is inside this open
set, machine will terminate with output $i$. Therefore, the
termination probability is the measure of this set, i.e., equals
$\lim_n p(i,n)$.
\end{proof}

\begin{problems}
\subsection*{Problems}

\leavevmode

1. A probabilistic machine without input terminates for all
possible coin tosses (there is no sequence of coin tosses that
leads to infinite computation). Prove that the computation time
is bounded by some constant (and machine can produce only finite
number of outputs).

2. Let $p_i$ be the probability of termination with output $i$
for some probabilistic machine and $\sum p_i=1$. Prove that all
$p_i$ are computable, i.e., for every given $i$ and for every
rational $\varepsilon>0$ we can find (algorithmically) an
approximation to $p_i$ with absolute error at most
$\varepsilon$.
\end{problems}

\section{A priori probability}

A sequence of real numbers $p_0,p_1,p_2,\dots$ is called an
\emph{lower semicomputable semimeasure} if there exists a
probabilistic machine (without input) that produces $i$ with
probability $p_i$. (As we know, $p_0,p_1,\dots$ is a lower semicomputable semimeasure if and only if $p_i$ is lower semicomputable and $\sum p_i \le 1$.)

\begin{theorem}
There exists a maximal lower semicomputable semimeasure $m$
\textup{(}maximality means that for every lower semicomputable semimeasure $m'$ there exists a constant $c$ such that $m'(i)\le cm(i)$ for all
$i$\textup{)}.
\end{theorem}

\begin{proof}
Let $M_0,M_1,\dots$ be a sequence of all probabilistic
machines without input. Let $M$ be a machine that starts by %??? which or that?
choosing a natural number $i$ at random (so that each outcome has
positive probability) and then emulates $M_i$. If $p_i$ is the
probability that $i$ is chosen, $m$ is the distribution on the
outputs of $M$ and $m'$ is the distribution on the outputs of
$M_i$, then $m(x) \ge p_i m'(x)$ for all $x$.
\end{proof}

The maximal lower semicomputable semimeasure is called \emph{a
priori probability}. This name can be explained as follows.
Imagine that we have a black box that can be turned on and
prints a natural number. We have no information about what is
inside. Nevertheless we have an ``a priori'' upper bound for
probability of the event ``$i$ appears'' (up to a constant
factor that depends on the box but not on $i$).

The same definition can be used for real-valued functions on
strings instead of natural numbers
(probabilistic machines produce strings; the sum $\sum
p(x)$ is taken over all strings $x$, etc.) --- in this way we may define \emph{discrete a priori probability} on binary strings. (There is another notion of a priori probability for strings, called \emph{continuous a priori probability}, but we do not consider it is this survey.)

\section{Prefix decompression}

The a priori probability is related to a special complexity measure
called \emph{prefix complexity}. The idea is that description is
self-delimited; the decompression program had to decide for itself
where to stop reading input. There are different versions of
machines with self-delimiting input; we choose one that is
technically convenient though may be not the most natural one.

A computable function whose inputs are binary strings is called
a \emph{prefix} function, if for every string $x$ and its prefix
$y$ at least one of the values $f(x)$ and $f(y)$ is undefined.
(So a prefix function cannot be defined both on a string and
its prefix or continuation.)

\begin{theorem}
There exists a prefix decompressor $D$ that is optimal among
prefix decompressors: for each computable prefix function $D'$
there exists some constant $c$ such that
        $$
\KS_{D}(x)\le \KS_{D'}(x)+c
        $$
for all $x$.
\end{theorem}

\begin{proof}
To prove a similar result for plain Kolmogorov complexity
we used
        $$
D(\overline {p}01 y)=p(y)
        $$
where $\overline{p}$ is a program $p$ with doubled bits and
$p(y)$ stands for the output of program $p$ with input $y$. This
$D$ is a prefix function if and only if all programs compute
prefix functions. We cannot algorithmically distinguish between
prefix and non-prefix programs (this is an undecidable problem).
However, we may convert each program into a prefix one in such a
way that prefix programs remain unchanged. Let us explain how this can be done.

Let
        $$
D(\overline {p}01 y)=[p](y)
        $$
where $[p](y)$ is computed as follows. We apply in parallel
$p$ to all inputs and get a sequence of pairs $\langle
y_i,z_i\rangle$ such that $p(y_i)=z_i$. Select a ``prefix''
subsequence by deleting all $\langle y_i,z_i\rangle$ such that
$y_i$ is a prefix of $y_j$ or $y_j$ is a prefix of $y_i$ for
some $j<i$. This process does not depend on $y$. To compute
$[p](y)$, wait until $y$ appears in the selected subsequence,
i.e. $y=y_i$ for a selected pair $\langle y_i,z_i\rangle$,
and then output~$z_i$.

The function $y\mapsto [p](y)$ is a prefix function for every $p$,
and if program $p$ computes a prefix function, then
$[p](y)=p(y)$.

Therefore, $D$ is an optimal prefix decompression algorithm.
\end{proof}

Complexity with respect to an optimal prefix decompression
algorithm is called \emph{prefix complexity} and denoted by
$\KP(x)$.

\section{Prefix complexity and length}

As we know, $\KS(x) \le |x|+O(1)$ (consider identity mapping as
decompression algorithm). But identity mapping is not a prefix one,
so we cannot use this argument to show that $\KP(x)\le |x|+O(1)$,
and in fact this is not true, as the following theorem shows.

\begin{theorem}
        $$
\sum_x 2^{-\KP(x)} \le 1.
        $$
\end{theorem}

\begin{proof}
For every $x$ let $p_x$ be the shortest description for $x$
(with respect to given prefix decompression algorithm). Then
$|p_x|=\KP(x)$ and all strings $p_x$ are incompatible. (We say
that $p$ and $q$ are compatible if $p$ is a prefix of $q$ or
vice versa.) Therefore, the intervals $\Omega_{p_x}$ are disjoint;
they have measure $2^{-|p_x|}=2^{-\KP(x)}$, so the sum does not
exceed $1$. 
\end{proof}

If $\KP(x)\le |x|+O(1)$ were true, then $\sum_x 2^{-|x|}$ would
be finite, but it is not the case (for each natural number $n$
the sum over strings of length $n$ equals $1$).

However, we can prove weaker lower bounds:

\begin{theorem}
        \begin{align*}
\KP(x) &\le 2|x|+O(1);\\
\KP(x) &\le |x|+2\log |x|+O(1);\\
\KP(x) &\le |x|+\log |x|+2\log\log|x|+O(1)\\
       &\ldots
        \end{align*}
\end{theorem}

\begin{proof}
The first bound is obtained if we use $D(\overline{x}01)=x$.
(It is easy to check that $D$ is prefix function.) The second one
uses
        $$
D(\overline{\bin(|x|)}01x)=x
        $$
where $\bin(|x|)$ is the binary representation of the length of
string $x$. Iterating this trick, we let
        $$
D(\overline{\bin(|\bin(|x|)|)}01\bin(|x|)x)=x
        $$
and get the third bound etc.
\end{proof}

Let us note that prefix complexity does not increase when we
apply algorithmic transformation: $\KP(A(x))\le \KP(x)+O(1)$ for
every algorithm $A$ (the constant in $O(1)$ depends on $A$). Let us take optimal decompressor (for plain
complexity) as $A$. We conclude that $\KP(x)$ does not exceed
$\KP(p)$ if $p$ is a description of $x$. Combining this
with theorem above, we conclude that $\KP(x)\le 2\KS(x)+O(1)$,
that $\KP(x) \le \KS(x)+2\log \KS(x)+O(1)$, etc.

In particular, the difference between plain and prefix complexity for $n$-bit strings is $O(\log n)$.

\section{A priori probability and prefix complexity}

We have now two measures for a string (or natural number) $x$. The 
a priori probability $m(x)$ measures how probable is to see $x$ as
an output of a probabilistic machine. Prefix complexity measures
how difficult is to specify $x$ in a self-delimiting way.
It turns out that these two measures are closely related.

\begin{theorem}
        $$
\KP(x)=-\log m(x)+O(1)
        $$
\end{theorem}

(Here $m(x)$ is a priori probability; $\log$ stands for binary
logarithm.)

\begin{proof}
The function $\KP$ is enumerable from above; therefore,
$x\mapsto 2^{-\KP(x)}$ is lower semicomputable. Also we know
that $\sum_x 2^{-\KP(x)}\le 1$, therefore $2^{-\KP(x)}$ is a
lower semicomputable semimeasure. Therefore, $2^{-\KP(x)}\le
cm(x)$ and
        $
\KP(x)\ge -\log m(x)+O(1).
        $
To prove that $\KP(x)\le -\log m(x)+O(1)$, we need the following
lemma about memory allocation.

Let the memory space be represented by $[0,1]$. Each memory
request asks for segment of length $1, 1/2, 1/4, 1/8$, etc. that
is properly aligned. Alignment means that for segment of length
$1/2^k$ only $2^k$ positions are allowed ($[0,2^{-k}],
[2^{-k},2\cdot 2^{-k}]$, etc.). Allocated segments should be
disjoint (common endpoints are allowed). Memory is never freed.
\smallskip

\textbf{Lemma}. \emph{For each computable sequence of requests
$2^{-n_i}$ such that $\sum 2^{-n_i}\le 1$ there is a computable
sequence of allocations that grant all requests.}

\begin{proof}
We keep a list of free space divided into segments of
size $2^{-k}$. Invariant relation: all segments are properly
aligned and have different size. Initially there is one free
segment of length $1$. When a new request of length $w$ comes,
we pick up the smallest segment of length at least $w$. This
strategy is sometimes called ``best fit'' strategy. (Note that
if the free list contains only segments of length
$w/2,w/4,\dots$, then the total free space is less than $w$, so
it cannot happen by our assumption.) If the smallest free segment of
length at least $w$ has length $w$, we simple allocate it (and
delete from the free list). If it has length $w'>w$, then we
split $w'$ into parts of size $w,w,2w,4w,\dots,w'/4,w'/2$ and
allocate the left $w$-segment putting all others in the free
list, so the invariant is maintained. 
\end{proof}

Reformulation of the lemma: \dots there is a computable sequence of
incompatible strings $x_i$ such that $|x_i|=n_i$. (Indeed, an
aligned segment of size $2^{-n}$ is $I_x$ for some string $x$
for length $n$.)
\smallskip

\textbf{Corollary}. \emph{For each computable sequence of requests
$2^{-n_i}$ such that $\sum 2^{-n_i}\le 1$ we have $\KP(i)\le n_i$.}

(Indeed, consider a decompressor that maps $x_i$ to $i$. Since
all $x_i$ are pairwise incompatible, it is a prefix function.)

Now we return to the proof. Since $m$ is lower semicomputable,
there exists a non-negative function $M:\langle x,k\rangle\mapsto M(x,k)$ of two arguments with rational values that is non-decreasing with
respect to the second argument such that $\lim_k M(x,k)=m(x)$.

Let $M'(x,k)$ be the smallest number in the sequence
$1,1/2,1/4,1/8,\dots,0$ that is greater than or equal to
$M(x,k)$. It is easy to see that $M'(x,k)\le 2M(x,k)$ and that
$M'$ is monotone.

We call pair $\langle x,k\rangle$ ``essential'' if $k=0$ or
$M'(x,k)>M'(x,k-1)$. The sum of $M'(x,k)$ for all essential
pairs with given $x$ is at most twice bigger than its biggest
term (because each term is at least twice bigger than the preceding
one), and its biggest term is at most twice bigger than $M(x,k)$
for some $k$. Since $M(x,k)\le m(x)$ and $\sum m(x)\le 1$, we
conclude that the sum of $M'(x,k)$ for all essential pairs $\langle
x,k\rangle$ does not exceed $4$.

Let $\langle x_i, k_i\rangle$ be a computable sequence of all
essential pairs. (We enumerate all pairs and select essential
ones.) Let $n_i$ be an integer such that
$2^{-n_i}=M'(x_i,k_i)/4$. Then $\sum 2^{-n_i}\le 1$.

Therefore, $\KP(i)\le n_i$. Since $x_i$ is obtained from $i$ by
an algorithm, we conclude that $\KP(x_i)\le n_i+O(1)$ for all
$i$. For a given $x$ one can find $i$ such that $x_i=x$ and
$2^{-n_i}\ge m_i/4$, so $n_i \le-\log m(x)+2$ and
$\KP(x)\le -\log m(x)+O(1)$.
\end{proof}

\section{Prefix complexity of a pair}

We can define $\KP(x,y)$ as prefix complexity of some code
$[x,y]$ of pair $\langle x,y\rangle$. As usual, different computable
encodings give complexities that differ at most by $O(1)$.

\begin{theorem}
        $$
\KP(x,y)\le \KP(x)+\KP(y)+O(1).
        $$
\end{theorem}

Note that now we do not need $O(\log n)$ term that was necessary for
plain complexity.

\begin{proof}
Let us give two proofs of this theorem using prefix
functions and a priori probability.

(1)~Let $D$ be the optimal prefix decompressor used in the
definition of $\KP$. Consider a function $D'$ such that
        $$
D'(pq)=[D(p), D(q)]
        $$
for all strings $p$ and $q$ such that $D(p)$ and $D(q)$ are
defined. Let us prove that this definition makes sense, i.e.,
that it does not lead to conflicts. Conflict happens if $pq=p'q'$
and $D(p), D(q), D(p'),D(q')$ are defined. But then $p$ and $p'$
are prefixes of the same string and are compatible, so $D(p)$
and $D(p')$ cannot be defined at the same time unless $p=p'$
(which implies $q=q'$).

Let us check that $D'$ is a prefix function. Indeed, if it is
defined for $pq$ and $p'q'$, and at the same time $pq$ is a prefix of $p'q'$, then
(as we have seen) $p$ and $p'$ are compatible and (since $D(p)$
and $D(p')$ are defined) $p=p'$. Then $q$ is a prefix of $q'$,
so $D(q)$ and $D(q')$ cannot be defined at the same time.

The function $D'$ is computable (for given $x$ we try all decompositions
$x=pq$ in parallel). So we have a prefix algorithm $D'$ such
that
        $\KS_D([x,y])\le \KP(x)+\KP(y)$
and therefore
        $\KP(x,y)\le\KP(x)+\KP(y)+O(1)$.
(End of the first proof.)

(2)~In terms of a priori probability we have to prove that
        $$
   m([x,y]) \ge \varepsilon m(x)m(y)
        $$
for some positive $\varepsilon$ and all $x$ and $y$.
Consider the function $m'$ determined by the equation
        $$
   m'([x,y]) = m(x)m(y)
        $$
($m'$ is zero for inputs that do not encode pairs of
strings). We have
        $$
\sum_z m'(z)=\sum_{x,y} m'([x,y])
=\sum_{x,y} m(x)m(y)= \sum_x m(x) \sum_y m(y)\le 1\cdot 1=1.
        $$
Function $m'$ is lower semicomputable, so $m'$ is a
semimeasure. Therefore, it is bounded by maximal semimeasure (up
to a constant factor). 
\end{proof}

A similar (but a bit more complicated) argument shows the equality
      $$
\KP(x,y)=\KP(x)+\KP(y|x,\KP(x))+O(1).
       $$

\section{Prefix complexity and randomness}

\begin{theorem}
        \label{randomness-criterion}
A sequence $x_0x_1x_2\dots$ is Martin-L\"of random if and only if
there exists some constant $c$ such that
        $$
\KP(x_0x_1\ldots x_{n-1})\ge n-c
        $$
for all $n$.
\end{theorem}

\begin{proof}
We have to prove that the sequence $x_0x_1x_2\dots$ is \emph{not}
random if and only if for every $c$ there exists $n$ such that
        $$
\KP(x_0x_1\ldots x_{n-1})< n-c.
        $$

(If-part)~A string $u$ is called (for this proof) $c$-defective if $\KP(u)<
|u|-c$. We have to prove that the set of all sequences that have
$c$-defective prefix for all $c$, is an effectively null set.
It is enough to prove that the set of all sequences that have
$c$-defective prefix for a given $c$ can be covered by intervals with total
measure $2^{-c}$.

Note that the set of all $c$-defective strings is enumerable
(since $\KP$ is enumerable from above). It remains to show that the sum
$\sum 2^{-|u|}$ over all $c$-defective $u$ does not exceed
$2^{-c}$. Indeed, if $u$ is $c$-defective, then by definition
$2^{-|u|}\le 2^{-c}2^{-KP(u)}$. On the other hand, the sum of
$2^{-\KP(u)}$ over all $u$ (and therefore over defective $u$)
does not exceed $1$.

(Only-if-part)~Let $N$ be the set of all non-random sequences. $N$ is
an effectively null set. For each integer $c$ consider a
sequence of intervals
        $$
\Omega_{u(c,0)},
\Omega_{u(c,1)},
\Omega_{u(c,2)},\dots
        $$
that cover $N$ and have total measure at most $2^{-2c}$.
Definition of effectively null sets guarantees that such a
sequence exists (and its elements can be effectively generated when $c$ is given).

For each $c,i$ consider the integer $n(c,i)=|u(c,i)|-c$. For a
given $c$ the sum $\sum_i 2^{-n(c,i)}$ does not exceed $2^{-c}$
(because the sum $\sum_i 2^{-|u(c,i)|}$ does not exceed
$2^{-2c}$). Therefore the sum $\sum_{c,i}2^{-n(c,i)}$ over all
$c$ and $i$ does not exceed $1$.

We would like to consider a semimeasure $M$ such that $M(u(c,i))=2^{-n(c,i)}$; however, it may happen that $u(c,i)$ coincide for different
pairs $c,i$. In this case we add the corresponding values, so the precise definition is
        $$
M(x)=\sum\{2^{-n(c,i)}\mid u(c,i)=x\}.
        $$
Note that $M$ is lower semicomputable, since $u$ and $n$ are computable
functions. Therefore, if $m$ is the universal semimeasure, we have
$m(x)\ge\varepsilon M(x)$, so $\KP(x)\le -\log M(x)+O(1)$, and
$\KP(u(c,i))\le n(c,i)+O(1)=|u(c,i)|-c+O(1)$.

If some sequence $x_0x_1x_2\dots$ belongs to the set $N$ of
non-random sequences, then it has prefixes of the form $u(c,i)$ for
all $c$, and for these prefixes the difference between length
and $\KP$ is not bounded.
\end{proof}

\section{Strong law of large numbers revisited}

Let $p,q$ be positive rational numbers such that $p+q=1$.
Consider the following semimeasure: a string $x$ of length $n$ with
$k$ ones and $l$ zeros has probability
        $$
\mu(x)=\frac{c}{n^2} p^{k} q^{l}
        $$
where constant $c$ is chosen in such a way that $\sum_n c/n^2
\le 1$. It is indeed a semimeasure (the sum over all strings $x$ is
at most $1$, because the sum of $\mu(x)$ over all strings $x$ of
given length $n$ is $1/n^2$; $p^kq^l$ is a probability to get
string $x$ for a biased coin whose sides have probabilities $p$ and $q$).

Therefore, we conclude that $\mu(x)$ is bounded by a priori
probability (up to a constant) and we get an upper bound
        $$
\KP(x) \le 2\log n + k (- \log p) + l (-\log q) +O(1)
        $$
for fixed $p$ and $q$ and for arbitrary string $x$ of length $n$ that
has $k$ ones and $l$ zeros. If $p=q=1/2$, we get the bound
$\KP(x)\le n+2\log n+O(1)$ that we already know. The new
bound is biased: If $p>1/2$ and $q<1/2$, then $-\log p < 1$ and
$-\log q >1$, so we count ones with less weight than zeros, and
new bound can be better for strings that have many ones and few
zeros.

Assume that $p>1/2$ and the fraction of ones in $x$ is greater that
$p$. Then our bound implies
        $$
\KP(x) \le 2\log n + np(- \log p) + nq(-\log q) +O(1)
        $$
(more ones make our bound only tighter). It can be rewritten as
        $$
\KP(x) \le n H(p,q) + 2\log n+ O(1)
        $$
where $H(p,q)$ is Shannon entropy for two-valued distribution
with probabilities $p$ and $q$:
        $$
H(p,q)=-p\log p -q \log q.
        $$
Since $p+q=1$, we have function of one variable:
        $$
H(p)=H(p,1-p)= -p\log p - (1-p)\log (1-p).
        $$
This function has a maximum at $1/2$; it is easy to check using
derivatives that $H(p)=1$ when $p=1/2$ and $H(p)<1$ when $p\ne
1/2$.
\smallskip

\textbf{Corollary}. \emph{For every $p>1/2$ there exist a constant $\alpha<1$ and
a constant $c$ such that
        $$
\KP(x) \le \alpha n + 2\log n +c
        $$
for each string $x$ where frequency of $1$s is at least $p$}.
\smallskip

Therefore, for every $p>1/2$, an infinite sequence of zeros and ones that has infinitely many prefixes with frequency of ones at least
$p$, is not Martin-L\"of random. This gives us a proof of
a constructive version of Strong Law of Large Numbers:

\begin{theorem}
Every Martin-L\"of random sequence $x_0x_1x_2\dots$ of zeros
and ones is balanced:
        $$
\lim_{n\to \infty} \frac{x_0+x_1+\ldots+x_{n-1}}{n} = \frac{1}{2}.
        $$
\end{theorem}

\begin{problems}
\subsection*{Problems}

\leavevmode

1. Let $D$ be a prefix decompression algorithm. Give a
direct construction of a probabilistic machine that outputs
$i$ with probability at least $2^{-K_D(i)}$.

2.$^*$ Prove that $\KP(x)\le \KS(x)+\KP(\KS(x))$

3. Prove that there exists an infinite sequence $x_0x_1\dots$
and a constant $c$ such that
        $$
\KS(x_0x_1\dots x_{n-1}) \ge n-2\log n + c
        $$
for all $n$.
\end{problems}

\section{Hausdorff dimension}

Let $\alpha$ be a positive real number. A set $X\subset\Omega$ of infinite bit sequences is called \emph{$\alpha$-null} if for every $\varepsilon>0$ there exists a set of strings $u_0,u_1,u_2,\ldots$ such that 

(1)~$X\subset \Omega_{u_0}\cup\Omega_{u_1}\cup\Omega_{u_2}\cup\ldots$;

(2)~$\sum_i 2^{-\alpha|u_i|} < \varepsilon$.

In other terms, we modify the definition of a null set: instead of the uniform measure $P(\Omega_u)=2^{-|u|}$ of an interval $\Omega_u$ we consider its \emph{$\alpha$-size} $(P(\Omega_u))^{\alpha}=2^{-\alpha|u|}$. For $\alpha>1$ we get a trivial notion: all sets are $\alpha$-null (one can cover the entire $\Omega$ by $2^N$ intervals of size $2^{-N}$, and $2^N\cdot 2^{-\alpha N}=1/2^{(\alpha-1)N}$ is small for large $N$). For $\alpha=1$ we get the usual notion of null sets, and for $\alpha<1$ we get a smaller class of sets (the smaller $\alpha$ is, the stronger condition we get).

For a given set $X\subset\Omega$ consider the infimum of $\alpha$ such that $X$ is an $\alpha$-null set. This infimum is called the \emph{Hausdorff dimension} of $X$. As we have seen, for the subsets of $\Omega$ the Hausdorff dimension is at most~$1$.

This is a classical notion but it can be constructivized in the same way as for null sets. A set $X\subset\Omega$ of infinite bit sequences is called \emph{effectively $\alpha$-null} if there is an algorithm that, given a rational $\varepsilon>0$, enumerates a sequence of strings $u_0,u_1,u_2,\ldots$ satisfying (1) and (2). The following result extends Theorem~\ref{maximal-null}:

\begin{theorem}
Let $\alpha>0$ be a rational number. Then there exists an effectively $\alpha$-null set $N$ that contains every effectively $\alpha$-null set.
\end{theorem}

\begin{proof}
We can use the same argument as for Theorem~\ref{maximal-null}: since $\alpha$ is rational, we can compute the $\alpha$-sizes of intervals with arbitrary precision, and this is enough to ensure that the sum of $\alpha$-sizes of a finite set of intervals is less than $\varepsilon$. (The same argument works for every computable $\alpha$.)
\end{proof}

Now we define \emph{effective Hausdorff dimension} of a set $X\subset \Omega$ as the infimum of $\alpha$ such that $X$ is an effectively $\alpha$-null set. It is easy to see that we may consider only rational $\alpha$ in this definition. The effective Hausdorff dimension cannot be smaller than the (classical) Hausdorff dimension, but may be bigger (see below).

We define the effective Hausdorff dimension of a point $\chi\in\Omega$ as the effective Hausdorff dimension of the singleton $\{\chi\}$. Note that there is no classical counterpart of this notion, since every singleton has Hausdorff dimension $0$.

For effectively null sets we have seen that this property of the set was essentially the property of its elements (all elements should be non-random); a similar result is true for effective Hausdorff dimension.

\begin{theorem}
      \label{dimension-supremum}
For every set $X$ its effective Hausdorff dimension equals the supremum of effective Hausdorff dimensions of its elements.
\end{theorem}

\begin{proof}
Evidently, the dimension of an element of $X$ cannot exceed the dimension of the set $X$ itself. On the other hand, if for some rational $\alpha>0$ all elements of $X$ have effective dimension less than $\alpha$, they all belong to the maximal effectively $\alpha$-null set, so $X$ is a subset of this maximal set, so $X$ is effectively $\alpha$-null set, and the effective dimension of $X$ does not exceed~$\alpha$.
\end{proof}

The criterion of Martin-L\"of randomness in terms of complexity (Theorem~\ref{randomness-criterion}) also has its counterpart for effective dimension. The previous result (Theorem~\ref{dimension-supremum}) shows that it is enough to characterize the effective dimension of singletons, and this can be done:

\begin{theorem}
The effective Hausdorff dimension of a sequence $\chi=x_0x_1x_2\ldots$ is equal to 
$$
    \liminf_{n\to\infty}\frac{\KP(x_0x_1\ldots x_{n-1})}{n}
$$
\end{theorem}

In this statement we use prefix complexity, but one may use the plain complexity instead (since the difference is at most $O(\log n)$ for $n$-bit strings).

\begin{proof}
If the $\liminf$ is smaller than $\alpha$, then $\KP(u)\le \alpha|u|$ for infinitely many prefixes of $\chi$. For the strings $u$ with this property we have 
$$2^{-\alpha|u|}\le m(u)$$
where $m$ is a priori probability, and the sum of $m(u)$ over all $u$ is bounded by~$1$. So we get a family of intervals that cover $\chi$ infinitely many times and have the sum of $\alpha$-sizes bounded by $1$. If we (1)~increase $\alpha$ a bit and consider some $\alpha'>\alpha$, and (2)~consider only strings $u$ of length greater than some large $N$, we get a family of intervals that cover $\chi$ and have small sum of $\alpha'$-sizes (bounded by $2^{(\alpha-\alpha')N}$, to be exact). This argument shows that the Hausdorff dimension of $\chi$ does not exceed the $\liminf$.

It remains to prove the reverse inequality. Assume that $\chi$ has effective Hausdorff dimension less than some (rational) $\alpha$. Then we can effectively cover $\chi$ by a family of intervals with arbitrarily small sum of $\alpha$-sizes. Combining the covers with sum bounded by $1/2, 1/4, 1/8,\ldots$, we get a computable sequence $u_0,u_1,u_2,\ldots$ such that 

(1)~intervals $\Omega_{u_0},\Omega_{u_1},\Omega_{u_2},\ldots$ cover $\chi$ infinitely many times;

(2)~$\sum 2^{-\alpha|u_i|} \le 1$.

The second inequality implies that $\KP(i)\le \alpha|u_i|+O(1)$, and therefore $\KP(u_i)\le \KP(i)+O(1)\le \alpha|u_i|+O(1)$. Since $\chi$ has infinitely many prefixes among $u_i$, we conclude that our $\liminf$ is bounded by $\alpha$. 
\end{proof}

This theorem implies that Martin-L\"of random sequences have dimension $1$ (it is also a direct consequence of the definition); it also allows us to construct easily a sequence of dimension $\alpha$ for arbitrary $\alpha\in(0,1)$ (by adding incompressible strings to increase the complexity of the prefix and strings of zeros to decrease it when needed). 

\begin{problems}
\section{Problems}

1. Let $k_n$ be average complexity of binary strings of length $n$:
        $$
k_n = \left[ \sum_{|x|=n} K(x)\right]/2^n.
        $$
Prove that $k_n=n+O(1)$ (i.e., $|k_n-n|<c$ for some $c$ and all $n$).

2. Prove that for a Martin-L\"of random sequence
$a_0a_1a_2a_3\dots$ the set of all $i$ such that $a_i=1$ is not
enumerable (there is no program that generates elements of this
set).

3. (Continued) Prove the same result for Mises--Church random
sequences.

4. String $x=yz$ of length $2n$ is incompressible: $\KS(x)\ge 2n$;
   strings $y$ and $z$ have length $n$. Prove that $\KS(y),\KS(z)\ge
   n-O(\log n)$. Can you improve this bound and show that $\KS(y),\KS(z)\ge n -O(1)$?

5. (Continued) Is the reverse statement (if $y$ and $z$ are
   incompressible, then $\KS(yz)=2n+O(\log n)$) true?

6. Prove that if $\KS(y|z)\ge n$ and $\KS(z|y)\ge n$ for strings $y$
   and $z$ of length $n$, then $\KS(yz)\ge 2n-O(\log n)$.

7. Prove that if $x$ and $y$ are strings of length $n$ and
   $\KS(xy)\ge 2n$, then the length of every common subsequence
   $u$ of $x$ and $y$ does not exceed $0.99n$. (A string $u$ is
   a subsequence of a string $v$ if $u$ can be obtained from $v$
   by deleting some terms. For example, $111$ is a subsequence
   of $010101$, but $1110$ and $1111$ are not.)

8. Let $a_0a_1a_2\dots$ and $b_0b_1b_2\dots$ be Martin-L\"of
   random sequences and $c_0c_1c_2\dots$ be a computable
   sequence. Can the sequence $(a_0\oplus b_0)(a_1\oplus
   b_1)(a_2\oplus b_2)\dots$ be non-random? (Here $a\oplus
   b$ denotes $a+b\bmod 2$.)
   The same question for $(a_0\oplus c_0)(a_1\oplus
   c_1)(a_2\oplus c_2)\dots$

9. True or false: $\KS(x,y)\le \KP(x)+\KS(y)+O(1)$?

10. Prove that for every $c$ there exists $x$ such that
    $\KP(x)-\KS(x)>c$.

11. Let $m(x)$ be a priori probability of string $x$. Prove that the
    binary representation of real number $\sum_x m(x)$ is a
    Martin-L\"of random sequence.

12. Prove that $\KS(x)+\KS(x,y,z)\le \KS(x,y)+\KS(x,z)+O(\log n)$ for
    strings $x,y,z$ of length at most $n$.

13. (Continued) Prove a similar result for prefix complexity
    with $O(1)$ instead of $O(\log n)$.
\end{problems}

\textbf{Acknowledgements}. This survey is based on the lecture notes of a course given in Uppsala University. The author's visit there was supported by STINT foundation. The author is grateful to all participants of Kolmogorov seminar (Moscow) and members of the ESCAPE group (Marseille, Montpellier).

The preparation of this survey was supported in part by the 
EMC ANR-09-BLAN-0164 and RFBR 12-01-00864 grant.

\small
\tableofcontents

\end{document}